\documentclass[runningheads,envcountsame]{llncs}
\usepackage{graphicx}
\usepackage{microtype}
\usepackage{ebproof}
\usepackage{stmaryrd}
\usepackage{mathtools,amssymb}
\usepackage{cmll}
\usepackage{xspace}
\usepackage{thmtools,thm-restate}
\usepackage{hyperref}
\usepackage{cleveref}
\usepackage{color}
\RequirePackage{ifthen}
\newcommand{\version}{0}
\newcommand{\commenti}{0} 
\newcommand{\condinc}[2]{\ifthenelse{\equal{\commenti}{0}}{#1}{**\violet{#2}} }
\newcommand{\SLV}[2]{\ifthenelse{\equal{\version}{0}}{#1}{ \RED{#2}}}

\newcommand{\Tot}{\ValScript}

\newcommand{\ValScript}{\mathsf{v}}

\newcommand{\Betav}{\beta_v}

\newcommand{\Derel}{\mathsf{d}}

\newcommand{\llsym}{\ell\ell}
\newcommand{\R}{\mathcal R}
\newcommand{\llsmall}{\textsc{l}}
\newcommand{\surf}{\textsc{s}}

\newcommand{\esssym}{\textsc{e}}
\newcommand{\op}{\mathbf{o}}

\newcommand{\DerSym}{\iota}
\newcommand{\Der}[1]{\DerSym\,#1}
\newcommand{\Derp}[1]{\DerSym(#1)}
\newcommand{\App}[2]{\langle #1 \rangle #2}
\newcommand{\La}[2]{\lambda #1. #2}
\newcommand{\Bang}[1]{\oc #1}
\newcommand{\Bangp}[1]{\oc(#1)}

\newcommand{\Nat}{{\mathbb N}}
\newcommand{\BangSet}{\Lambda_\oc}
\newcommand{\BangSetOp}{\Lambda_{\oc \OpSet}}
\newcommand{\BangSetCtx}{{\CtxSet_\oc}}
\newcommand{\BangSetOpCtx}{{\CtxSet_{\oc\OpSet}}}

\newcommand{\BoxSet}{\oc\Lambda_\oc}
\newcommand{\BoxSetOp}{\oc\Lambda_{\oc,\OpSet}}
\newcommand{\ValSet}{\mathsf{Val}}
\newcommand{\VarSet}{\mathsf{Var}}

\newcommand{\LambdaVal}{\Lambda_v}
\newcommand{\LambdaOpVal}{{\LambdaOp}_v}
\newcommand{\CtxSet}{\mathcal{C}}
\newcommand{\LambdaOpCtx}{{\CtxSet_\OpSet}}
\newcommand{\LambdaCtx}{{\CtxSet}}
\newcommand{\OpSet}{\mathcal{O}}
\newcommand{\LambdaOp}{\Lambda_\OpSet}

\newcommand{\Fv}[1]{\mathsf{fv}(#1)}
\newcommand{\Sub}[2]{\{#1/#2\}}

\newcommand{\Rule}{\rho}

\newcommand{\Root}[1]{\mapsto_{#1}}
\newcommand{\To}[1]{\to_{#1}}

\newcommand{\MRevTo}[1]{{\,}_{#1}^*\!\!\leftarrow}

\newcommand{\ToBang}{\To{\oc}}
\newcommand{\ToVal}{\To{\ValScript}}

\newcommand{\ToTot}{\To{\Tot}}

\newcommand{\ToBeta}{\To{\beta}}
\newcommand{\ToBetav}{\To{\Betav}}
\newcommand{\ToOp}{\To{\op}}

\newcommand{\ToRule}{\To{\Rule}}
\newcommand{\ToRuleAt}[1]{\To{\Rule:#1}}

\newcommand{\ToBetavAt}[1]{\To{\Betav:#1}}
\newcommand{\ToBangAt}[1]{\ToTotInd{#1}}
\newcommand{\ToOpAt}[1]{\To{\op:#1}}

\makeatletter
\newcommand{\uset}[3][0ex]{%
	\mathrel{\mathop{#3}\limits_{
			\vbox to#1{\kern-6\ex@
				\hbox{$\scriptstyle#2$}\vss}}}}
\makeatother

\newcommand{\red}{\mathrel{\rightarrow}}

\newcommand{\essred}{\uset{{\esssym~}}{\red}}
\newcommand{\nessred}{\uset{\neg\esssym}{\red}}

\newcommand{\sred}{\uset{{\surf~}}{\red}}
\newcommand{\nsred}{\uset{\neg\surf}{\red}}

\newcommand{\sredx}[1] {\mathrel{\sred{}_{\mkern-6mu#1}}}
\newcommand{\nsredx}[1]{\mathrel{\nsred{}_{\mkern-6mu#1}}}

\newcommand{\sredBang}{\mathrel{\sredx{{\Tot}}}}
\newcommand{\nsredBang}{\mathrel{\nsredx{{\Tot}}}}

\newcommand{\llred}{\uset{\llsmall~}{\red}}
\newcommand{\nllred}{\uset{\neg \llsmall~}{\red}}

\newcommand{\llredx}[1]  {\mathrel{\llred{}_{\mkern-6mu#1}}}
\newcommand{\nllredx}[1]{\mathrel{\nllred{}_{\mkern-6mu#1}}}
\newcommand{\llredb}  {\mathrel{\llredx{\beta}}}
\newcommand{\nllredb}{\mathrel{\nllredx{\beta}}}

\newcommand{\llredo}  {\mathrel{\llredx{\oplus}}}

\newcommand{\llredbv}  {\mathrel{\llredx{\Betav}}}

\newcommand{\llredBang}  {\mathrel{\llredx{\Tot}}}
\newcommand{\nllredBang}{\mathrel{\nllredx{\Tot}}}
\newcommand{\llredRule}  {\mathrel{\llredx{\Rule}}}
\newcommand{\nllredRule}{\mathrel{\nllredx{\Rule}}}
\newcommand{\llredOp}{\mathrel{\llredx{\op}}}
\newcommand{\nllredOp}{\mathrel{\nllredx{\op}}}

\newcommand{\Parallel}{\mathrel{\Rightarrow}}
\newcommand{\Parallelx}[1]{\Parallel_{#1}}

\newcommand{\ParallelBang}{\Parallelx{\Tot}}
\newcommand{\ParallelBangAt}[1]{\Parallelx{\Tot:#1}}

\newcommand{\ParallelNotEss}{\mathrel{\uset{{\lnot\esssym}}{\Parallel}}}

\newcommand{\ParallelNotLl}{\mathrel{\uset{{\lnot\llsmall}}{\Parallel}}}

\newcommand{\ParallelNotLlx}[1]{\mathrel{\ParallelNotLl{}_{\mkern-6mu#1}}}

\newcommand{\ParallelNotLlBang}{\ParallelNotLlx{\Tot}}

\newcommand{\ParallelBangInd}[1]{\overset{#1}{\Parallel}_{\Tot}}
\newcommand{\ParallelBangIndLong}[1]{\overset{#1}{\Parallelx{\Tot}}}

\newcommand{\var}{x}
\newcommand{\varTwo}{y}
\newcommand{\varThree}{z}

\newcommand{\Tm}{T}
\newcommand{\TmTwo}{S}
\newcommand{\TmThree}{R}
\newcommand{\TmFour}{Q}
\newcommand{\TmFive}{P}

\newcommand{\tm}{t}
\newcommand{\tmTwo}{s}
\newcommand{\tmThree}{r}
\newcommand{\tmFour}{q}

\newcommand{\ImTmOne}{M}
\newcommand{\ImTmTwo}{N}

\newcommand{\ImVal}{U}
\newcommand{\LVal}{V}
\newcommand{\LTm}{\Tm}
\newcommand{\LTmTwo}{\TmTwo}
\newcommand{\LTmThree}{\TmThree}
\newcommand{\LTmFour}{\TmFour}

\newcommand{\Hole}[1]{\langle #1 \rangle}
\newcommand{\HoleCtx}{\Hole{\cdot}}

\newcommand{\Ctx}{\mathbf{C}}

\newcommand{\ctx}{\mathtt{c}}
\newcommand{\WCtx}{\mathtt{W}}

\newcommand{\Ctxp}[1]{\Ctx\Hole{#1}}

\newcommand{\LCtx}{\mathbf{c}}
\newcommand{\LCtxp}[1]{\LCtx\Hole{#1}}
\newcommand{\LCtxTwo}{\LCtx'}

\newcommand{\Cbn}[1]{{#1^\mathsf{cbn}}}
\newcommand{\Cbv}[1]{{#1^\mathsf{cbv}}}
\newcommand{\CbN}{CbN\xspace}
\newcommand{\CbV}{CbV\xspace}

\newcommand{\InvV}[1]{#1^\dagger}

\newcommand{\Defeq}{\coloneqq}
\newcommand{\Eqdef}{\eqqcolon}

\newcommand{\Size}[1]{|#1|}
\newcommand{\SizeP}[2]{\Size{#1}_{#2}}

\newcommand{\lev}[1]{\ell(#1)}
\newcommand{\levCbn}[1]{\ell^\textup{\CbN}(#1)}
\newcommand{\levCbv}[1]{\ell^\textup{\CbV}(#1)}

\newcommand{\llCbn}[1]{\llsym^\textup{\CbN}(#1)}
\newcommand{\llCbv}[1]{\llsym^\textup{\CbV}(#1)}

\newcommand{\llBangR}[1]{\llbRed{}{\R}{#1}}
\newcommand{\llBang}[1]{\llbRed{}{}{#1}}
\renewcommand{\deg}[1]{\llsym(#1)}

\newcommand{\llb}[2]{\llsym^{#1}(#2)}
\newcommand{\llbRed}[3]{\llsym^{#1}_{#2}(#3)}

\newcommand{\Ie}{\textit{i.e.}\xspace}
\newcommand{\ie}{\textit{i.e.}\xspace}
\newcommand{\Ih}{\textit{i.h.}\xspace}

\newcommand{\Resp}{\textnormal{resp.}\xspace}

\newcommand\Isom\simeq

\renewcommand{\LVal}{v}
\renewcommand{\LTm}{\tm}
\renewcommand{\LTmTwo}{\tmTwo}
\renewcommand{\LTmThree}{\tmThree}
\renewcommand{\LTmFour}{\tmFour}

\renewcommand{\ToBang}{\To{\Derel}}

\newcommand{\ToDer}{\ToBang}
\newcommand{\ToDerNorm}{\twoheadrightarrow_\Derel}

\newcommand{\ToTotInd}[1]{\To{\Tot:#1}}

\newcommand{\ToDerInd}[1]{\To{\Derel:#1}}
\newcommand{\ToBetaInd}[1]{\To{\beta:#1}}

\newcommand{\llredDer}  {\mathrel{\llredx{\Derel}}}
\newcommand{\nllredDer}{\mathrel{\nllredx{\Derel}}}

\renewcommand{\ValScript}{\beta_\oc}
\renewcommand{\VarSet}{\mathcal{V}\!ar}

\newcommand{\valu}{box\xspace}

\renewcommand{\Cbn}[1]{{#1^\mathsf{n}}}
\renewcommand{\Cbv}[1]{{#1^\mathsf{v}}}

\renewcommand{\WCtx}{\mathtt{G}}

\newcommand\retract\vartriangleleft
\newcommand\lam{\ensuremath{\lambda}}

\let\Gamma\varGamma
\let\Delta\varDelta
\let\Theta\varTheta
\let\Lambda\varLambda
\let\Xi\varXi
\let\Pi\varPi
\let\Sigma\varSigma
\let\Upsilon\varUpsilon
\let\Phi\varPhi
\let\Psi\varPsi
\let\Omega\varOmega

\usepackage[normalem]{ulem}

\newcommand{\clean}{0}
\ifthenelse{\equal{\clean}{0}}
{
\newcommand{\old}[1]{{\textcolor{red}{#1}}}
\newcommand{\CF}[1]{\vskip 8pt \textcolor{violet}{*CF: #1 *}}
\newcommand{\memo}[1]{\todo{MEMO:#1}}
}{
\newcommand{\old}[1]{}
\newcommand{\CF}[1]{}
\newcommand{\memo}[1]{}
}

\usepackage{amsmath}
\usepackage{url}
\usepackage{mathtools}

\usepackage{color}
\usepackage{graphicx}
\usepackage{ebproof}
\usepackage{cmll}
\usepackage{stmaryrd}
\usepackage{proof}
\usepackage{wrapfig}
\usepackage{amssymb}
\usepackage{bm}

\usepackage{mathtools}
\usepackage{xspace}
\usepackage{cancel}

\newcommand{\xredx}[2] {\mathrel{{\uset{#1}{\red}}{}_{\mkern-3mu#2}}}

\newcommand{\xbackredx}[2] {\mathrel{{\uset{#1}{\leftarrow}}{}_{\mkern-3mu#2}}}
\newcommand{\xrevredx}{\xbackredx}

\newcommand{\eq}{\,=\,}
\newcommand{\eqdef}{\,:=\,}

\newcommand{\RulesSet}{\mathsf{Rules}}

\newcommand{\hole}[1]{\langle #1\rangle}

\newcommand{\cc}{\textsf {C}}

\renewcommand{\ll}{\textsc {l}}

\newcommand{\lo}{\textsc {lo}}

\newcommand{\llev}[1]{{\ell \ell}(#1)}

\newcommand{\cbv}{{\mathtt{cbv}}}
\newcommand{\cbn}{{\mathtt{cbn}}}
\newcommand{\betav}{{\beta_v}}

\newcommand{\ex}{\mathsf {e}}
\renewcommand{\int}{\mathsf{i}}

\newcommand{\head}{\mathsf{h}}

\newcommand{\id}{\iota}

\newcommand{\betab}{!\beta}

\newcommand{\rredbb}{\mapsto_{\bbeta}}

\newcommand{\tms}{s}
\newcommand{\tmu}{u}
\newcommand{\tmr}{r}
\newcommand{\tmp}{p}

\renewcommand{\to}{\xrightarrow{}}

\newcommand{\tob}{\to_\beta}
\newcommand{\tobv}{\to_{\betav}}

\newcommand{\PRed}{\Rightarrow}
\newcommand{\iPRed}{\uset{\int}{\PRed}}

\newcommand{\redx}[1]  {\mathrel{\red{}_{\mkern-8mu#1}}}

\newcommand{\ered}{\uset{\ex}{\red}}
\newcommand{\ired}{\uset{\int}{\red}}

\newcommand{\nered}{\uset{\neg \ex~}{\red}}

\newcommand{\hred}{\uset{\head}{\red}}

\newcommand{\hredx}[1]  {\mathrel{\hred{}_{\mkern-8mu#1}}}

\newcommand{\hredb}  {\mathrel{\hredx{\beta}}}

\newcommand{\hredo}  {\mathrel{\hredx{\oplus}}}

\newcommand{\llredc}  {\mathrel{\llredx{\gamma}}}
\newcommand{\nllredc}{\mathrel{\nllredx{\gamma}}}

\newcommand{\nllreda}{\mathrel{\nllredx{\alpha}}}

\newcommand{\sredbb}  {\mathrel{\sredx{\bbeta}}}

\newcommand{\llredbb}  {\mathrel{\llredx{\bbeta}}}
\newcommand{\nllredbb}{\mathrel{\nllredx{\bbeta}}}

\newcommand{\redbv}{\red_{\beta_v}}

\newcommand{\redb}{\rightarrow_{\beta}}
\newcommand{\eredx}[1]  {\mathrel{\ered{}_{\mkern-8mu#1}}}
\newcommand{\iredx}[1]  {\mathrel{\ired{}_{\mkern-8mu#1}}}

\newcommand{\bbeta}{!\beta}
\newcommand{\redbb}{\rightarrow_{\bbeta}}

\newcommand{\reda}  {\red_{\alpha}}
\newcommand{\ereda}  {\mathrel{\eredx{\alpha}}}

\newcommand{\ireda}  {\mathrel{\iredx{\alpha}}}

\newcommand{\redc}  {\red_{\gamma}}
\newcommand{\eredc} {\mathrel {\eredx{\gamma}}}
\newcommand{\iredc}  {\mathrel{\iredx{\gamma}}}

\newcommand{\lored}{\uset{\lo}{\red}}

\newcommand{\loredb}  {\mathrel{\lored{}_{\mkern-8mu\beta}}}

\newcommand{\loredx}[1]  {\mathrel{\lored{}_{\mkern-8mu#1}}}

\newcommand{\redo}{\rightarrow_{\oplus}}
\newcommand{\nllredo}{\nllredx{\oplus}}

\newcommand{\parmark}{\circ\mkern -1mu}

\newcommand{\makepar}[1]{~\parmark \mkern-16mu #1}

\newcommand{\iparred}{{\makepar \ired}}

\newcommand{\PLambda}{\Lambda_\oplus}

\spnewtheorem*{example*}{Example}{\itshape}{\itshape}

\spnewtheorem{fact}[theorem]{Property}{\bfseries}{\itshape}
\spnewtheorem{notation}[theorem]{Notation}{\bfseries}{\itshape}

\newcommand{\eg}{\emph{e.g.}\xspace}
\newcommand{\ih}{\emph{i.h.}\xspace}

\usepackage{xcolor}

\newcommand{\todo}[1]{{ \color{red}{#1}}}

\newcommand{\RED}[1]{{\color{red}{#1}}}
\newcommand{\pink}[1]{{\color{magenta}{#1}}}

\newcommand{\violet}[1]{{\color{violet}{#1}}}

\newcommand{\rredc}{\mapsto_{\gamma}}
\newcommand{\rreda}{\mapsto_{\alpha}}

\newcommand{\subs}[2]{ \{#2/#1\} }

\renewcommand{\AA}{A}

\newcommand{\LP}[2]{\mathtt{SP(#1,#2)}}
\newcommand{\F}[2]{\mathtt{Fact(#1,#2)}}

\usepackage{ebproof}

\newcommand{\esym}{{\mathtt e}}

\newcommand{\toe}{\to_{\esym}}

\renewcommand{\pink}[1]{}

\renewcommand{\bbeta}{\Tot}
\renewcommand{\betab}{\Tot}

\renewcommand{\gamma}{\rho}
\renewcommand{\alpha}{\xi}

\renewcommand{\surf}{0}
\renewcommand{\WCtx}{\mathbf{S}}

\begin{document}
\title{Factorization in Call-by-Name and Call-by-Value Calculi via Linear Logic}
%

\author{Claudia Faggian\inst{1} \and
Giulio Guerrieri\inst{2}}
\authorrunning{C. Faggian \and G. Guerrieri}
%
\institute{Université de Paris, IRIF, CNRS, F-75013 Paris, France \and University of Bath, Department of Computer Science, Bath, UK
}
\maketitle              
\begin{abstract}
	In each variant of the $\lambda$-calculus, factorization and normalization are two key-properties that show how 
	results are computed.
	
Instead of proving factorization/normalization 
for the call-by-name (\CbN) and call-by-value (\CbV) variants separately, we prove them only once, for the bang calculus (an extension of the $\lambda$-calculus inspired by linear logic  and subsuming \CbN and \CbV), and then we transfer 
the result via translations, obtaining factorization/normalization  for  \CbN and \CbV.

The approach is robust: it still holds when extending the calculi with operators 
and extra rules to model  some additional computational features.

\end{abstract}

\section{Introduction}
\label{sect:intro}

The  $\lambda$-calculus is  the  model of computation underlying functional programming languages and proof assistants.
Actually there are many $\lambda$-calculi, depending on the \emph{evaluation mechanism} (for instance, call-by-name and call-by-value---\CbN and \CbV for short) 
and \emph{computational features} that the calculus aims to model. 

In $\lambda$-calculi, a rewriting relation formalizes computational steps in program execution, and normal forms 
are the results of computations. 
In each calculus, a key question is to define a \emph{normalizing strategy}: How to compute a result? Is there a reduction strategy which is guaranteed to output a result, if any exists?

Proving that a calculus admits a normalizing strategy is complex, and many techniques have been developed. 
A well-known method first proves \emph{factorization} \cite{Barendregt84,Takahashi95,HirokawaMiddledorpMoser15,AccattoliFaggianGuerrieri19}.
Given a calculus with a rewriting relation $\to$, a strategy $\llred \subseteq \to$  \emph{factorizes} if  $\to^* \subseteq \llred^* \cdot \nllred^*$ ($\nllred$ is the dual of $\llred$), \ie  any reduction sequence can be rearranged so as to perform first $\llred$-steps and then the other steps.
If, moreover, the strategy satisfies some “good properties”,
we can conclude that the strategy is normalizing.
Factorization is important also because it is commonly used as a building block in the proof of other properties of the \emph{how-to-compute} kind. 
For instance, \emph{standardization}, which generalizes factorization: every reduction sequences can be rearranged according to a predefined order between redexes.

\paragraph{Two for One.} 
Quoting Levy \cite{Levy99}: \emph{the existence of two separate paradigms} (\CbN and \CbV) is troubling because to prove a certain property---such as factorization or normalization---for both systems \emph{we always need to do it twice}.

The \emph{first aim} of our paper is to develop a technique for deriving factorization for both the 
\CbN \cite{Barendregt84} and 
\CbV \cite{Plotkin75} $\lam$-calculi as corollaries of a \emph{single} factorization theorem, and similarly for normalization. A key tool in our study is the \emph{bang calculus} \cite{EhrhardGuerrieri16,GuerrieriManzonetto18}, a calculus inspired by linear logic in which 
\CbN and \CbV embed.

\paragraph{The Bang Calculus.}

The bang calculus  is a variant of the $\lambda$-calculus where an operator $\oc$ plays the role of a marker for non-linear management: duplicability and discardability.
The bang calculus is nothing but Simpson's linear $\lambda$-calculus \cite{Simpson05} without linear abstraction, or the untyped version of the implicative fragment of Levy's Call-by-Push-Value \cite{Levy99}, as first observed by Ehrhard \cite{Ehrhard16}. 

The motivation to study the bang calculus is to have a general framework where both \CbN and \CbV $\lambda$-calculi can be simulated, via two distinct \emph{translations} inspired by Girard's embeddings \cite{Girard87} of the intuitionistic arrow into linear logic.
So, a certain property can be studied in the bang calculus and then automatically transferred to the \CbN and \CbV settings by translating back.

This approach has so far mainly be exploited semantically \cite{Levy06,Ehrhard16,EhrhardGuerrieri16,GuerrieriManzonetto18,ChouquetTasson20,BucciarelliKesnerRiosViso20}, 
	but can be used it also  to study  operational  properties \cite{GuerrieriManzonetto18,SantoPintoUustalu19,FaggianRonchi}.
In this paper, we push forward this operational direction.

\paragraph{The Least-Level Strategy.}
We study a strategy from the literature of linear logic \cite{CarvalhoPF11}, namely \emph{least-level reduction} $\llred$, which fires a redex at minimal level---the \emph{level} of a redex $\tmThree$ is the number of $\oc$ under which the redex appears.

We prove that the least-level reduction factorizes and normalizes in the bang calculus, and then we transfer the same results to \CbN and \CbV $\lam$-calculi (for suitable definitions of least-level in \CbN and \CbV), by exploiting properties of their translations into the bang calculus. A single proof suffices.
It is two-for-one! Or even better, three-for-one.

The rewriting study of the least level strategy in the bang calculus is based on simple techniques for factorization and normalization we developed recently 
with Accattoli \cite{AccattoliFaggianGuerrieri19}, which simplify and generalize Takahashi's method \cite{Takahashi95}.

\paragraph{Subtleties of the Embeddings.}  Transferring factorization and normalization results via translation is highly non-trivial, 
\eg in CPS translations \cite{Plotkin75}. This applies also to   
transferring least-level factorizations from the bang calculus to the \CbN and \CbV $\lambda$-calculi.
To transfer the property smoothly, the translations should preserve levels and normal forms, which is delicate, in particular for \CbV. 
The embedding of \CbV into the bang calculus defined in \cite{GuerrieriManzonetto18,SantoPintoUustalu19} 
does not preserve levels and normal forms (see \Cref{rmk:least-level-normal}). 
As a consequence, the \CbV translation studied in \cite{GuerrieriManzonetto18,SantoPintoUustalu19} cannot be used to derive 
least-level factorization or \emph{any} 
normalization result in a \CbV setting from the corresponding result in the bang calculus.

Here we adopt the refined \CbV embedding of Bucciarelli et al. \cite{BucciarelliKesnerRiosViso20} which does preserve levels and normal forms. While the preservation of normal forms is already stressed in \cite{BucciarelliKesnerRiosViso20}, the preservation of levels is proved here for the first time, and it is based on non-trivial properties of the embedding.

%

\paragraph{Beyond pure.} 
Our \emph{second aim} is to show that the developed technique for the joined factorization and normalization of \CbN and \CbV via the bang calculus is \emph{robust}. We do so, by studying extensions of all three calculi with operators (or, in general, with extra rules) which model  some additional computational features, such as non-deterministic  or probabilistic choice. We then show that the technique scales up smoothly, under mild assumptions on the extension. 

\paragraph{A Motivating Example.} Let us illustrate our approach on a simple case, which we  will use as running example.
De' Liguoro' and Piperno's \CbN non-deterministic $\lam$-calculus $\PLambda^\cbn$ \cite{deLiguoroP95} extends  the CbN $\lam$-calculus with an operator $\oplus$ whose reduction  models  \emph{non-deterministic choice}:
$\oplus (\tm, \tms)$ 
rewrites to either  $\tm$ or  $\tms$. 
It admits a standardization result, from which if follows that the  leftmost-outermost reduction strategy (noted $\loredx{\beta\oplus}$) is \emph{complete}: if  $\tm$ has a  \textit{normal form} $\tmu$ then  	$  \tm \loredx{\beta\oplus}^*  \tmu$.
In \cite{deLiguoro91}, de' Liguoro considers also a \CbV variant $ \PLambda^\cbv $, extending with an operator $\oplus$  the \CbV $\lam$-calculus. One may prove standardization and completeness---again---from scratch, even though  the proofs are similar.

The  approach we propose here is to work  in the bang calculus enriched with the operator $\oplus$, it is denoted by ${\BangSet}_{\oplus}$.
We show that the calculus satisfies \emph{least-level factorization}
from which it follows that the least-level strategy is \emph{complete}, \ie if   $\tm$ has a   \textit{normal form} $\tmu$, then 
$ \tm \llredx{\bbeta\oplus}^* \tmu $. The translation then  guarantees that  analogous results hold also in $\PLambda^\cbn$ and $\PLambda^\cbv$.

\paragraph{The Importance of Being Modular.} The bang calculus with operators is actually a general formalism for several calculi, one calculus for each kind of computational feature modeled by operators. 
Concretely, the reduction $\red$ consists of $\ToTot$ (which subsumes \CbN $\ToBeta$ and \CbV $\ToBetav$) and other reduction rules $\ToRule$.

We decompose the proof  of factorization of $\red$ in  modules, by using the \emph{modular approach} recently introduced by the authors together with Accattoli \cite{AccattoliFaggianGuerrieri21}. 

The key module is the  least-level factorization of $\redbb$, because it is where the higher-order comes into play---this is done, once for all.
Then, we consider a generic reduction rule $\ToRule$ to add to $\ToTot$. Our general result is that if $\ToRule$ has `good properties' and interacts well with $\ToTot$ (which amounts to an easy test, combinatorial in nature), then we have least-level factorization for $\ToTot \cup \ToRule$. 

Putting all together, when $\ToRule$ is instantiated to a  concrete reduction (such as $\redo$), the user of our method only has to verify 
a simple test (namely \Cref{prop:test_ll}), to conclude that $\ToTot \cup \redo$ has least-level factorization. In particular factorization for $\ToTot$ is a ready-to-use black box the user need not to worry about---our  proof is robust enough to hold  whatever the other rules are. Finally, the embedding automatically give least-level factorization for the corresponding \CbV and \CbN calculi.
In \Cref{sec:case_study}, we illustrate our method on this example.

\paragraph{Subtleties of the Modular Extensions.}
In order to adopt the modular approach presented in \cite{AccattoliFaggianGuerrieri21} we need to deal  
with  an important difficulty which appears when dealing with normalizing strategies and that it is not studied in \cite{AccattoliFaggianGuerrieri21}.


A  normalizing   strategies select   the redex to fire usually   through a property such as being a \emph{least level} redex or being the \emph{leftmost-outermost} (shortened to LO) redex---normalizing strategies are \emph{positional}.

The problem is that the---in general---  if $\red =\redb\cup\redx{\Rule}$, then $\loredx{}$ reduction is not the union of  $\loredx{\beta}$ and  $\loredx{\Rule}$. I.e., the normalizing strategy of the compound system is not obtained putting together the normalizing strategies of the components.
Let us explain the issue on our running example $\redx{\beta \oplus}$, in the familiar case   of leftmost-outermost reduction.

\renewcommand{\lid}{\mathsf{I}}

\begin{example}\label{ex:issue}
	Let us first consider head reduction with respect to  $\beta$ (written $\hredb $) and with respect to  $\beta\oplus$ (written $\hredx{\beta\oplus} $).
	Consider the term $\tms=(\lid\lid)(x \oplus y)$, where $\lid=\lam x.x$. The subterm $\lid\lid$ (which is a $\beta$-redex) is in head position whenever we consider the reduction $\redb$ or its extension $\redx{\beta \oplus}$. So $\tms \hredb \lid(x\oplus y)$ and $\tms \hredx{\beta\oplus} \lid(x\oplus y)$.  Conversely, 
	given  $\tm=(x \oplus y)(\lid\lid)$ the head position is occupied by $ (x\oplus y) $, which is a $\oplus$-redex, but not a $\beta$-redex. Therefore,
	$(\lid\lid) $ is not the head-redex in $\tm$, neither for $\beta$ nor for $\beta\oplus$.
	Otherwise stated:  \[\hredx{\beta\oplus}~ = ~\hredb \cup \hredo.\]
	
	In contrast, if  we consider leftmost-outermost reduction $\lored$, which reduces a redex in the leftmost-outermost position, it is easy to see that \[\loredx{\beta\oplus} ~\not= ~\loredb \cup \loredx{\oplus}.\]
	Consider again  the term $\tm=(x \oplus y){(\lid\lid)}$. 
	Since $ (x\oplus y) $ is not a $\beta$-redex, $(\lid\lid)$ is the leftmost redex for  $\redb$. Instead, $(\lid\lid)$ is not  the $\lo$-redex for  $\redx{\beta \oplus}$ (here  the leftmost redex is $(x\oplus y)$). 
	 So  $\tm \loredb (x\oplus y)\lid$ but  $\tm \not\loredx{\beta\oplus} (x\oplus y)\lid$. 
\end{example}

The least-level factorization for $\redbb$, $\redb$, and $\redbv$ we prove here is robust enough to make it ready to be used as a module in a larger proof, where it may 
  combine with operators and other rules.
The key point is to define the least-level reduction from the very beginning as a reduction firing a redex at minimal level with respect to a general set of redexes (containing $\Tot$, $\beta$ or $\Betav$, respectively), so that it is ``ready'' to be extended with other reduction rules (see \Cref{sec:ll}).

\paragraph{Proofs.} All proofs are available in \url{https://www.irif.fr/~giuliog/fact.pdf}
%

\section{Background in Abstract Rewriting}\label{sec:background}

An (\emph{abstract}) \emph{rewriting system}, \cite[Ch. 2]{Terese03}  is a pair $(\AA, \to)$ 
consisting of a  set $A$ and a binary 
relation $\to \subseteq \AA\times \AA$ (called reduction) whose pairs are written  $t \to s$ and called \emph{steps}. 
A \emph{$\to$-sequence} from $\tm$ is a sequence of $\to$-steps. 
As usual, 
$\red^*$ (resp. $\red^=$) denotes the transitive-reflexive (resp. reflexive) closure of $\red$.

A relation $\red$ 
is 
\emph{confluent} if  $s \MRevTo{} r\red^{*} t$ implies
$s \red^{*}u \MRevTo{} t$ for some $u$.
We  say that $u$ is $\red$-\emph{normal} (or a $\red$-normal form) if there is no $t$ such that 
$u\red t$. 

In general, a term may or may not reduce to a normal form. 
If it does, not all reduction sequences necessarily lead to normal form. 
A term is 
\emph{weakly} or \emph{strongly normalizing}, depending on if it  may or must reduce to normal form.
If a term $\tm$ is strongly normalizing, any choice of steps will eventually lead to a normal form. 
However, if $\tm$ is weakly normalizing, how do we compute a normal form? This is the problem tackled by \emph{normalization}:  by repeatedly performing \emph{only specific  steps},  a normal form will be computed, provided that $\tm $ can reduce to~any.

%

A \emph{strategy} $\toe \ \subseteq \ \to$ is a way to control 
that in a term there are different possible choices of reduction.
A  \emph{normalizing strategy}  for $\red$, is a reduction strategy which, given a term $\tm$, is guaranteed to reach its $\red$-normal form, if any exists (a key tool to show that certain terms are not $\red$-normalizable).

\begin{definition}[Normalizing and complete strategy]
	\label{def:strategy}
A reduction $\ered \,\subseteq \, \red$ is a  \emph{strategy for} $\red$  if it has the same normal forms as $\red$. 
A strategy $\ered$ for~$\red$~is:
\begin{itemize}
 \item \emph{complete} if 
 $t\ered^*u$ whenever $t\to^*u$ with   $u$ $\red$-normal;
 \item \emph{normalizing} if \emph{every} maximal $\ered$-sequence from $t$ ends in a normal form, whenever  $t\to^*u$ for some $\red$-normal form $u$. 
\end{itemize}
\end{definition}
Note that if  the strategy $\ered$ is complete and  \emph{deterministic} (\ie for every $t\in \AA$, $t\ered s$ for at most one $s\in \AA$), then  $\ered$ is a normalizing strategy for $\red$.

\begin{definition}[Factorization]
Let $(A,\to)$ be a rewriting system with $\red \eq \ered \cup  \ired $.  The relation   $\red$  
satisfies  \emph{$\ex$-factorization}, written $\F{\ered}{\ired}$, if
	\begin{equation}\tag{\textbf{Factorization}}
		\F{\ered}{\ired}: \quad (\ered \cup  \ired)^*~ \subseteq ~\ered^* \cdot \ired^*  
	\end{equation}	
\end{definition}	

\subsubsection{Proving Normalization.}
Factorization provides a simple technique to establish that a strategy is normalizing.
\begin{lemma}[Normalization \cite{AccattoliFaggianGuerrieri19}]\label{prop:abs_normalization} Let $\red \eq \ered  \cup  \nered$, and $\ered$ be a \mbox{strategy for $\red$.}
	
	The strategy $\ered$  is \emph{complete} for $\red$ if  the following hold:
	\begin{enumerate}
		\item \emph{Persistence:} If $\tm \nered \tm'$  then $\tm'$ is not normal. 
		\item  \emph{Factorization}:   $\tm\to^* \tmu$ implies $ \tm\ered^* \!\cdot\! \nered^*\tmu$.
	\end{enumerate}

 The strategy $\ered$  is \emph{normalizing} for $\red$ if it is complete and: 
 \begin{enumerate}\setcounter{enumi}{2}
		\item 
		\emph{Uniformity:} all weakly $\ered$-normalizing terms are strongly $\ered$-normalizing.
 \end{enumerate} 
\end{lemma}

A sufficient condition for uniform normalization and confluence \mbox{is the following}:

\begin{fact}[Newman \cite{Newman42}]\label{fact:diamond}
	A reduction  is \emph{quasi-diamond} if   ($\tm_1\leftarrow \tm \rightarrow \tm_2$) implies ($\tm_1=\tm_2$ or $\tm_1\rightarrow \tmu \leftarrow \tm_2$ for some $\tmu$).
	If $\red$ is quasi-diamond 
	then $\red$ is uniformly normalizing and confluent.
		
		
\end{fact}

%
%
%
%
%
%

\subsubsection{Proving Factorization.}\label{sec:Hindley}

Hindley\cite{HindleyPhD} first noted  that a local 
property implies 
factorization. Let $\red \, = \, \ered \cup \ired$.
We say that  $\ired$ \emph{strongly postpones} after $\ered$,
if
\begin{equation}\label{eq:SP}\tag{\textbf{Strong Postponement}}
\LP{\ered}{\ired}:
 \qquad	\ired \cdot \ered ~\subseteq~\ered^*\cdot  \ired^=
\end{equation}


\begin{lemma}[Hindley \cite{HindleyPhD}]
\label{l:SP} 
	$\LP{\ered}{\ired}$ implies $\F{\ered}{\ired}$.
\end{lemma}

Strong postponement can rarely be used \emph{directly}, because several  interesting 
reductions---including $\beta$-reduction---do not  satisfy it. 
However, it is
at the heart of Takahashi's 
method \cite{Takahashi95} to prove head  factorization of $\ToBeta$, via the following immediate property that can be used also to prove other factorizations (see \cite{AccattoliFaggianGuerrieri19}).
\begin{fact}[Characterization of factorization]
	\label{fact:fact}
	Factorization	$\F \ered \ired$ holds 
	\emph{if and only if} 
	there is a reduction $\iparred$ such that $\iparred^* \, = \ \ired^*$ and $\LP{\ered}{\iparred}$.
	
\end{fact}

The core of Takahashi's method 
\cite{Takahashi95} is to 
introduce a  relation $\iPRed$, called \emph{internal parallel reduction}, which verifies 
the hypotheses above.
We will follow a similar path  in  \Cref{sect:factorization}, to prove \textit{least-level} factorization.

\subsubsection{Compound systems: proving  factorization in a modular way.}
\label{sec:modular}

In this paper, we will consider compound systems  that are  obtained by   extending the 
$\lam$-calculus with extra rules to model  advanced features. 

In an abstract setting, let us consider a   rewrite system $(A,\red)$ where $\red \eq\reda \cup \redc$. 
Under which condition $\red$ admits factorization, assuming that both $\reda$ and $\redc$ do?
To deal with this question,  a technique   for proving factorization  for \emph{compound systems}  in a 
\emph{modular} way  has been introduced in \cite{AccattoliFaggianGuerrieri21}.
 The approach  can be seen as an  analogous for factorization of the classical 
technique for confluence based on Hindley-Rosen lemma \cite{Barendregt84}: 
if   $\reda,\redc$ are $\ex$-factorizing reductions, their union  $\reda \cup \redc$ also is, 
 provided that two \textit{local} conditions of commutation hold.

\begin{lemma}[Modular factorization \cite{AccattoliFaggianGuerrieri21}]\label{thm:modular}
	Let  $\reda \eq \ereda \cup \ireda$ and $\redc \eq \eredc \cup \iredc$ be  $\ex$-factorizing relations.
	Let $\ered  \eqdef \ereda \cup\eredc$, and $\ired  \eqdef  \ireda \cup\iredc$. 
	The  union
	$\reda\cup \redc$  fulfills factorization 
	$\F{\ered}{\ired}$  if the following  swaps hold
	\begin{equation}\label{eq:LS}\tag{\textbf{Linear Swaps}}
\ireda \cdot   \eredc  \ \subseteq\  \eredc \cdot \reda^*     \quad\text{ and }\quad \iredc \cdot   \ereda  \ \subseteq\  \ereda \cdot \redc^* 
	\end{equation}
	
\end{lemma}

The subtlety here 
is to set $\ereda$ and $\eredc$ 
so that $\ered\eq  \eredc \cup \iredc$. As already shown in \Cref{sect:intro},  when dealing with normalizing strategies one needs extra~care.



\newcommand{\oterm}{$\op$-redex\xspace}
\newcommand{\oterms}{$\op$-redexes\xspace}

\renewcommand{\VarSet}{\mathsf{Var}}
\renewcommand{\ValSet}{\mathsf{Val}}
\renewcommand{\LambdaVal}{\ValSet}
\renewcommand{\LambdaOpVal}{\LambdaVal}
\renewcommand{\BangSetOpCtx}{\BangSetCtx}
\renewcommand{\LambdaOpCtx}{\LambdaCtx}

\section{ \texorpdfstring{$\lambda$}{lambda}-calculi: \CbN, \CbV, and bang}
\label{sect:lambda-calculi}

We present here a generic syntax for $\lambda$-calculi, possibly containing  operators.  
All the variants of the $\lambda$-calculus we shall study use this language. 
We assume some familiarity with the $\lam$-calculus, and  refer to \cite{Barendregt84,HindleySeldin86} for  details.

Given a countable set $\VarSet$ of variables, denoted by $\var, \varTwo, \varThree, \dots$, \emph{terms} and \emph{values} (whose sets are denoted by $\LambdaOp$ and $\LambdaOpVal$, respectively) are defined 
as~follows: 
\begin{align*}
	\LTm, \LTmTwo, \LTmThree &\Coloneqq \LVal \mid \LTm\LTmTwo \mid \op(\LTm_1, \dots, \LTm_k) \ \text{ \emph{Terms}:~ } \LambdaOp 
	&&&&&
	\LVal &\Coloneqq \var \mid \La{\var}{\LTm} \ \text{ \emph{Values}:~ }\LambdaOpVal
\end{align*}
where 
$\op$ ranges over a set $\OpSet$ of function symbols called \emph{operators},
each one with its own arity $k \in \Nat$.
If the operators  are $\op_1,\dots, \op_n$, the set of terms is  
indicated as $\Lambda_{\op_1...\op_n}$. When the set $\OpSet$ of operators is empty,  the 
calculus is called \emph{pure}, and the sets of terms is denoted by $\Lambda$. Otherwise, the calculus is \emph{applied}.

Terms are identified up to renaming of bound variables, where abstraction is the only binder. 
We denote by $\LTm\Sub{\LTmTwo}{\var}$ the capture-avoiding substitution of $\LTmTwo$ for the free occurrences of $\var$ in $\LTm$.
\emph{Contexts} (with exactly one hole $\HoleCtx$) are generated by the grammar below, and $\LCtx\Hole{\LTm}$ stands for the term obtained from the context $\LCtx$ by replacing the  hole with the term $\LTm$ (possibly capturing free variables).
\begin{align*}
				\LCtx & \Coloneqq \Hole{\cdot}   \mid \LTm\LCtx \mid \LCtx \LTm \mid \La{\var}\LCtx \mid \op(\LTm_1, \dots, \LCtx,\dots, \LTm_k)
				& \text{\emph{Contexts:~}} \LambdaCtx	
	\end{align*}

Let $ \Rule$ be a binary relation on $\LambdaOp$; we call it \emph{$\Rule$-rule} and denote it also by $\Root{\Rule}$, writing $\tm \Root{\Rule} \tm'$ rather than $(\tm,\tm')\in \rho$.
	A $\Rule$-\emph{reduction step}  $\To{\Rule}$ is 
the contextual closure of $\Rule$. 
 Explicitly,
$\LTm \ToRule \LTm'$
holds if  $\LTm = \LCtxp{\LTmThree}$ and $\LTm' = \LCtxp{\LTmThree'}$ for some context $\LCtx$ with $ \LTmThree \Root{\Rule} \LTmThree'$.
The term $\LTmThree$ is called a $\Rule$-\emph{redex}. The set of 
$\Rule$-\emph{redexes} is denoted~by~$\R_{\Rule}.$
 
Given a set of rules $\RulesSet$, the relation $\red = \bigcup_\Rule \ToRule$ ($\Rule \in \RulesSet$)  can equivalently be defined as the contextual closure of $\mapsto = \bigcup_\Rule \Root{\Rule}$. 

\condinc{}{
Given a binary relation $\Rule$ on $\LambdaOp$, called \emph{$\Rule$-rule},
 a $\Rule$-\emph{reduction step}  $\To{\Rule}$ is 
the contextual closure of $\Rule$.  
Explicitly, 
$\LTm \ToRule \LTm'$
holds if    $ (\LTmThree,\LTmThree')\in \Rule$,  $\LTm = \LCtxp{\LTmThree}$, and $\LTm' = \LCtxp{\LTmThree'}$,  for some context $\LCtx$.

The term $\LTmThree$ is called a $\Rule$-\emph{redex}. The set of all $\Rule$-\emph{redex} is denoted $\R_{\Rule}.$
We write the pair   $ (\LTmThree,\LTmThree')\in \Rule $ also as  $\LTmThree \Root{\Rule}\LTmThree'$. 

Given a set of rules $\RulesSet$, the relation $\red = \bigcup \ToRule$ ($\Rule \in \RulesSet$)  can equivalently be defined as the contextual closure of $\bigcup \Rule$. We also write $\red =\{\ToRule \mid \Rule \in \RulesSet \}$.
}




\pink{\paragraph{General properties of the contextual closure.}
We  recall a basic but key  property of contextual closure.
We say that $\tm$ and $\tm'$ have \emph{the same shape} if  both terms are an 
	application (resp. an abstraction, a variable, a constant,  or a term of shape  $\op (\tmp_1, ...,\tmp_k)$ ).
\begin{fact}[Shape preservation]\label{fact:shape} 
	Assume $\tm=\LCtxp{\tmr}\redx{\Rule} \LCtxp {\tmr'}=\tm'$ and that the  context   $\ctx$ is \emph{non-empty}. Then $\tm$ and $\tm'$ have the same shape.
\end{fact}
That is,  if a step $\redc$ is obtained by closure under \emph{non-empty context}, then it preserve the shape of the term.

Please notice that we will often write $\rredc$ to indicate the step $\redc$ which is obtained by \emph{empty contextual closure}.
}

\subsection{Call-by-Name and Call-by-Value $\lam$-calculi}

\paragraph{Pure \CbN and Pure \CbV $\lambda$-calculi.}
	\label{ex:cbn-cbv-calculi}
	The \emph{pure call-by-name} (\CbN for short) $\lambda$-calculus \cite{Barendregt84,HindleySeldin86} is  $(\Lambda, \redb)$, the set of terms $\Lambda$ together  with the $\beta$-reduction  $\redb$, 
defined as the contextual closure of the  usual  $\beta$-rule, which we recall in \eqref{eq:rule-beta} below.

The \emph{pure call-by-value} (\CbV for short) $\lambda$-calculus \cite{Plotkin75} is the set $\Lambda$ endowed with the reduction $\ToBetav$, defined as the contextual closure of the  $\Betav$-rule 
in \eqref{eq:rule-betav}.

\noindent
	\begin{minipage}{0.38\linewidth}
		\begin{equation}
		\label{eq:rule-beta}
		\textup{\CbN: }  \, (\La{\var}\LTm)\LTmTwo \Root{\beta} \LTm\Sub{\LTmTwo}{\var}
		\end{equation}
	\end{minipage}
	\quad
	\begin{minipage}{0.56\linewidth}
		\begin{equation}
		\label{eq:rule-betav}
		\textup{\CbV: } \,	(\La{\var}{\LTm})\LVal \mapsto_{\Betav} \LTm\Sub{\LVal}{\var} \mbox{ \ with } \LVal \!\in\! \LambdaVal
		\end{equation}
	\end{minipage} 

\paragraph{\CbN and \CbV $\lambda$-calculi.}
A \CbN (resp. \CbV) $\lambda$-calculus   is the set  of terms endowed with a reduction $\To{} $ which extends $\redb$ (resp. $\redbv$).

In particular, the   \emph{applied}  setting with operators (when $\OpSet \neq \emptyset$) models in the $\lam$-calculus richer computational features,
 allowing $\op$-reductions as the contextual closure of  $\op$-rules of the form  $\op (\LTm_1, \dots, \LTm_k) \Root{\op} \LTmTwo$.

\begin{example} [Non-deterministic $\lam$-calculus]\label{ex:NDext} Let  $\OpSet =\{\oplus\}$ where $\oplus$ is a binary operator; let  $\redo$  be the contextual  
	closure of   the (non-deterministic) rule below:
	\[	\oplus (\tm_1,\tm_2)\mapsto_{\oplus} \tm_1 \quad \text{ and }   \quad   \oplus (\tm_1,\tm_2)\mapsto_{\oplus} \tm_2\]
	
The \emph{non-deterministic \CbN $ \lam $-calculus} 
 $\PLambda^\cbn=(\Lambda_{\oplus},\redx{\beta\oplus})$ is the set $\Lambda_{\oplus}$ 
 with the reduction $\redx{\beta\oplus} \ = \ \tob \cup \red_{\oplus} $. 
 The \emph{non-deterministic \CbV $ \lam $-calculus}
 $\PLambda^\cbv=(\Lambda_{\oplus},\redx{\betav\oplus})$ is the set $\Lambda_{\oplus}$ 
 with the reduction $\redx{\betav\oplus} \ = \ \tobv \cup \red_{\oplus} $.
\end{example}

\subsection{Bang calculi}
\label{sect:bang-calculus}

The bang calculus \cite{EhrhardGuerrieri16,GuerrieriManzonetto18} is a variant of the $\lambda$-calculus inspired by linear logic. 
An operator $\oc$ plays the role of a marker for duplicability and discardability.
Here we allow also the presence of operators other than $\Bang$, ranging over a set $\OpSet$.
So, terms and contexts of the bang calculus (denoted by capital letters) are: 
\begin{align*}
\Tm, \TmTwo, \TmThree &\Coloneqq \var \mid \La{\var}{\Tm} \mid \Tm\TmTwo \mid \Bang{\Tm} \mid \op(\Tm_1, \dots, \Tm_k) 
&
\text{\emph{Terms:~}}\BangSetOp
\\
\Ctx & \Coloneqq \Hole{\cdot}  \mid \La{\var}\Ctx \mid \Tm\Ctx \mid \Ctx \Tm \mid \Bang{\Ctx} \mid \op(\Tm_1, \dots, \Ctx,\dots, \Tm_k)
& 
\text{\emph{Contexts:~}}\BangSetOpCtx
\end{align*}
Terms of the form $\Bang{\Tm}$ are called \emph{boxes} and their set is denoted by $\BoxSetOp$.
When there are no operators other than $\oc$ (\ie $\OpSet = \emptyset$), the sets of terms,  boxes and contexts are denoted by $\BangSet$, $\BoxSet$ and $\BangSetCtx$, respectively.
This syntax can be expressed in the one of \Cref{sect:lambda-calculi}, where $\oc$ is an unary operator called \emph{bang}. 

\paragraph{The pure bang calculus.}
The  \emph{pure} bang calculus  $(\BangSet, \redbb)$ is the set of  terms $\BangSet$  endowed with reduction $\ToTot$, the closure under contexts in $\BangSetOpCtx$ of the \emph{$\Tot$-rule}: 
	\begin{equation}
			(\La{\var}{\Tm})\,\Bang{\TmTwo} \Root{\Tot} \Tm \Sub{\TmTwo}{\var}
	\end{equation}

Intuitively, in the bang calculus the bang-operator $\oc$ marks the only 
terms that can be erased and duplicated.
Indeed, a \emph{$\beta$-like redex} $(\La{\var}{\Tm}){\TmTwo}$ can be fired by $\Root{\Tot}$ only when its argument $\TmTwo$ is a \valu, \Ie$\TmTwo = \Bang{\TmThree}$: 
if it is so, the content $\TmThree$ of the box $\TmTwo$ (and not $\TmTwo$ itself) replaces any free occurrence of $\var$ in $\Tm$.\footnotemark
\footnotetext{Syntax and reduction rule of the bang calculus 
	follow \cite{GuerrieriManzonetto18}, which 
	is slightly different from 
	\cite{EhrhardGuerrieri16}. 
	Unlike \cite{GuerrieriManzonetto18} (but akin to \cite{SantoPintoUustalu19}), here we do not use  $\Der{\!}$ (aka $\mathsf{der}$) as a primitive, since $\Der{\!}$ and its associated rule $\Root{\Derel}$ can be simulated, see \Cref{ex:identity}~and~\eqref{eq:der}.}

A proof of confluence of $\Tot$-reduction $\ToTot$ is in \cite{GuerrieriManzonetto18}.
\newcommand{\bid}{I}
\begin{notation}
	\label{note:terms}
	We use the following notations to denote some notable terms.
	\begin{align*}
 \Der{\!} := \La{\var}{\var} &&  \delta:=\La{\var}{\var}{\var} && \bid := \La{\var}{\Bang{\var}} && \Delta := \La{\var}{\var\,\Bang{\var}}.
 \end{align*} 
 \end{notation}
 \begin{remark}[Notable terms]\label{ex:delta}\label{ex:identity}
 	The 
 	term $\bid = \La{\var}{\Bang{\var}}$ plays the role of the identity in the bang calculus: $\bid \, \Bang{\Tm} \ToTot \Bangp{\var\Sub{\Tm}{\var}} = \Bang{\Tm}$ for any term $\Tm$.
 	Instead, the term $\Der{\!} = \La{\var}{\var}$, when applied to a box $\Bang{\Tm}$, opens the box, \Ie returns its content $\Tm$: $\Der{}\Bang{\Tm} \ToTot \var\Sub{\Tm}{\var} = \Tm$.
 	Finally,  $\Delta\, \Bang{\Delta} \ToTot \Delta\, \Bang{\Delta} \ToTot \dots$ is a diverging term.
 \end{remark}

%


\paragraph{A bang calculus.}
A \emph{bang calculus}  $(\BangSetOp, \To{})$    is the set $\BangSetOp$ of terms endowed with a reduction $\To{} $ which extends $ \ToTot$.
In this paper  we shall consider  calculi where 
 $\to$ contains  $\ToTot$ and  $\op$-reductions $\ToOp$ ($\op\in \OpSet $) defined from $\op$-rules of the form  $\op (\Tm_1, \dots, \Tm_k) \Root{\op} \TmTwo$, and possibly other rules.
  So,  $\red \,=\bigcup_{\Rule}\ToRule  (\Rule\in \RulesSet)$, with $\RulesSet \supseteq \{\Bang{\beta}, \op \mid \op\in \OpSet\}$.
 We set $\redx{\OpSet} \,= \bigcup_{\op\in \OpSet}\redx{\op}$.

\subsection{\CbN and \CbV translations into the bang calculus}
\label{subsect:translations}

Our  motivation to study the bang calculus is 
to have a general framework 
where both \CbN \cite{Barendregt84} and 
\CbV \cite{Plotkin75} $\lambda$-calculi 
can be embedded, via two distinct translations. 
Here we  show how these translations work.
We extend the simulation results
in \cite{GuerrieriManzonetto18,SantoPintoUustalu19,BucciarelliKesnerRiosViso20} 
for the pure case to the case with operators (\Cref{thm:embedding}).

Following \cite{BucciarelliKesnerRiosViso20}, the \CbV translation defined here differs from \cite{GuerrieriManzonetto18,SantoPintoUustalu19} in the application case.
\Cref{sect:embedding} will show why this optimization is crucial.

\emph{\CbN} and \emph{\CbV} \emph{translations} are two maps $\Cbn{(\cdot)} \colon \LambdaOp \to \BangSetOp$ and $\Cbv{(\cdot)} \colon \LambdaOp \to \BangSetOp$, respectively, translating terms of the $\lambda$-calculus into terms of the bang~calculus: 
\begin{align*}
\Cbn{\var} &\Defeq \var 						& \ \Cbn{(\La{\var\,}{\LTm})} &\Defeq \La{\var\,}{\Cbn{\LTm}} 
& \Cbn{(\op(\LTm_1, \dots, \LTm_k))} &\Defeq \op(\Cbn{\LTm_1}, \dots, \Cbn{\LTm_k})
& \ \Cbn{(\LTm\LTmTwo)} &\Defeq \Cbn{\LTm} \,{\Bang{{\Cbn{\LTmTwo}}}} \,; \\ 
\Cbv{\var} &\Defeq \Bang{\var} & \	\Cbv{(\La{\var\,}{\LTm})} &\Defeq \Bang{(\La{\var}{\Cbv{\LTm}})} 
& \Cbv{(\op(\LTm_1, \dots, \LTm_k))} &\Defeq \op(\Cbv{\LTm_1}, \dots, \Cbv{\LTm_k})
& \ \Cbv{(\LTm\LTmTwo)} &\Defeq 
\begin{cases}
\Tm \,\Cbv{\LTmTwo} & \text{if } \Cbv{\LTm} = \Bang{\Tm}
\\
(\Der{\Cbv{\LTm}}){\Cbv{\LTmTwo}} &\text{otherwise}.
\end{cases}
\end{align*}

\begin{example}
	\label{ex:delta-translated}
	Consider the $\lambda$-term  $\omega \Defeq \delta\delta$: 
	then, $\Cbn{\delta} = \Delta$, $ \Cbv{\delta} = \Bang{ \Delta}$ and 
	$\Cbn{\omega} = \Delta\,{\Bang{\Delta}} = \Cbv{\omega}$ 
	($\delta$ and $\Delta$ are defined in \Cref{note:terms}).
	The $\lambda$-term $\omega$ is diverging in \CbN and \CbV $\lambda$-calculi, and so is $\Cbn{\omega} = \Cbv{\omega}$ in the bang calculus, see \Cref{ex:delta}.
\end{example}

For any term $\LTm \in \LambdaOp$, $\Cbn{\LTm}$ and $\Cbv{\LTm}$ are just different decorations of $\LTm$ by means of the bang-operator $\oc$ (recall that $\Der{\!} = \La{\var}{\var}$). 
The translation $\Cbn{(\cdot)}$ puts the argument of any application into a box: in \CbN 
any term is duplicable or discardable.
On the other hand, only \emph{values} (\Ie abstractions and variables) are translated by $\Cbv{(\cdot)}$ into boxes, as they are the only terms duplicable or discardable in~\CbV. 

As in \cite{GuerrieriManzonetto18,SantoPintoUustalu19}, we prove that the \CbN translation $\Cbn{(\cdot)}$ (\Resp \CbV translation $\Cbv{(\cdot)}$) from the pure \CbN (\Resp \CbV) $\lambda$-calculus into the bang calculus is \emph{sound} and \emph{complete}: it maps $\beta$-reductions (\Resp $\Betav$-reductions) of the $\lambda$-calculus into $\Tot$-reductions of the bang calculus, and conversely $\Tot$-reductions\,---\,when  restricted to the image of the translation\,---\,into $\beta$-reductions (\Resp $\Betav$-reductions). 
The same holds if we consider any $\op$-reduction for operators.

In the simulation, $\ToDer$ denotes the contextual closure of the rule:
	\begin{equation}
	\label{eq:der}
		\Der{\Bang{\Tm}} \Root{\Derel} \Tm  \quad \text{(this is nothing but } (\La{\var}{\var}) \Bang{\Tm} \Root{\Tot} \Tm \text{)}
	\end{equation}
	Clearly, $\ToDer \, \subseteq \, \ToTot$ (\Cref{ex:identity}).
	We write $\Tm \ToDerNorm \TmTwo$ if $\Tm \ToDer^* \TmTwo$ and $\TmTwo$ is $\Derel$-normal.  

\begin{restatable}[Simulation of \CbN and \CbV]{proposition}{embedding}
	\label{thm:embedding}
	Let
	$\LTm \in \LambdaOp$ and $\op \in \OpSet$. 
	\begin{enumerate}
		\item\label{p:embedding-cbn} 
		\emph{\CbN soundness:} If $\LTm \ToBeta \LTm'$ then 
		$\Cbn{\LTm} \ToTot \Cbn{{\LTm'}}$.  
		If $\LTm \ToOp \LTm'$ then 
		$\Cbn{\LTm} \ToOp \Cbn{{\LTm'}}$.
		\\
		\emph{\CbN completeness:} 
		If $\Cbn{\LTm} \ToTot \TmTwo$ then $\TmTwo = \Cbn{{\LTm'}}$ and $\LTm \ToBeta \LTm'$, for some $\LTm' \in \LambdaOp$.
		If $\Cbn{\LTm} \ToOp \TmTwo$ then $\TmTwo = \Cbn{{\LTm'}}$ and $\LTm \ToOp \LTm'$, for some $\LTm'\in \LambdaOp$.
		\item\label{p:embedding-cbv}
		\emph{\CbV soundness:} If $\LTm \ToBetav \LTm'$ then $\Cbv{\LTm} \ToVal\ToDer^= \Cbv{{\LTm'}}$ 
		with $\Cbv{{\LTm'}}$ $\Derel$-normal.
		If $\LTm \ToOp \LTm'$ then $\Cbv{\LTm} \ToOp\ToDer^= \Cbv{{\LTm'}}$ 
		with $\Cbv{{\LTm'}}$ $\Derel$-normal.
		\\
		\emph{\CbV completeness:} 
		If $\Cbv{\LTm} \ToVal \ToDerNorm \TmTwo$ 
		then $\Cbv{\LTm} \ToVal \ToDer^= \TmTwo$ with $\TmTwo = \Cbv{{\LTm'}}$ and $\LTm \ToBetav \LTm'$, for some $\LTm' \in \LambdaOp$.
		If $\Cbv{\LTm} \ToOp \ToDerNorm \TmTwo$ 
		then $\Cbv{\LTm} \ToOp \ToDer^= \TmTwo$ with $\TmTwo = \Cbv{{\LTm'}}$ and $\LTm \ToOp \LTm'$, for some $\LTm' \in \LambdaOp$.
	\end{enumerate}
\end{restatable}

\pink{Note that
\emph{one step} of $\beta$-reduction corresponds exactly, via $\Cbn{(\cdot)}$, to \emph{one step} of  
$\Tot$-reduction, and vice-versa; 
\RED{and \emph{one step} of $\Betav$-reduction corresponds exactly, via $\Cbv{(\cdot)}$, to \emph{one step} of  $\ValScript$-reduction, and vice-versa.} The same holds for $\op$-reduction.}

\begin{example}
	\label{ex:simulation}
	Let $\LTm = ((\La{\varThree}{\varThree})\var)\varTwo$ and $\LTm' = \var\varTwo$. Then 
	 $\LTm \ToBeta \LTm'$ while $\Cbn{\LTm} = ((\La{\varThree}{\varThree})\Bang{\var})\Bang{\varTwo} \ToTot \var\, \Bang{\varTwo} = \Cbn{{\LTm'}}$;
	and  $\LTm \ToBetav \LTm'$ while $\Cbv{\LTm} = (\Derp{(\La{\varThree}{\Bang{\varThree}})\Bang{\var}})\Bang{\varTwo} \ToVal (\Der{\Bang{\var}})\Bang{\varTwo} \ToDer \var\, \Bang{\varTwo} = \Cbv{{\LTm'}}$.
\end{example}

\section{The least-level strategy}\label{sec:ll}
\newcommand{\good}{good\xspace}

The bang calculus $\BangSet$ has a natural normalizing 
strategy, issued  by linear logic (where it was first used in \cite{CarvalhoPF11}), namely  the \emph{least-level reduction}. It reduces only redexes at \emph{least level}, where the   \emph{level} of a redex $R$ in a term $\Tm$ is  the number of boxes $\oc$ 
 in which  $R$ is nested. 
 
Least-level reduction is easily extended to  a general bang calculus $(\BangSetOp, \To{})$.
The level of a redex $R$ is then the number of boxes $\oc$  and  operators $\op$ in which  $R$ is nested; 
intuitively, least-level reduction fires a redex which is \emph{minimally nested}. 

Below, we formalize the  reduction in a way that is independent of the specific shape of the redexes, and even of  specific definition of level one chooses.
The interest of least-level reduction is in the properties it satisfies. All our developments will rely on such properties, rather than the specific definition of least level.

In this section, 
$\red=\bigcup_\Rule \ToRule,  \text{for } \Rule\in \RulesSet $ a set of rules.  We write $\R = \bigcup_\Rule \R_{\Rule}$ for   the set of  \emph{all}  redexes.

\subsection{Least-level reduction in bang calculi}\label{sec:ll_explicit}\label{sec:ll_def}

%
%

 The \emph{level} of an occurrence of redex $R$ in a term $\Tm$ is a measure of its depth. 
Formally, we indicate the \emph{occurrence of   a subterm} $R$ in $\Tm$ with  the context $\Ctx$ such that $\Ctxp{R} = \Tm$.
Its  level   then corresponds to the \emph{level} $\lev{\Ctx}$ of the hole in $\Ctx$.
The definition of \emph{level} in a bang calculus $\BangSetOp$ is formalized  as follows.
\begin{equation}
	\label{eq:level-bang}
	\begin{gathered}
	\lev{\Hole{\cdot}}  = 0 \qquad
	\lev{\La{\var}\Ctx} = \lev {\Ctx} \qquad
	\lev {\Ctx \Tm} = \lev {\Ctx} \qquad
	\lev {\Tm \Ctx} = \lev {\Ctx}
	\\
	\lev {\Bang{\Ctx}} = \lev {\Ctx} + 1 \qquad
	\lev {\op (\dots, \Ctx, \dots)} = \lev{\Ctx} +1
	\end{gathered}
	\end{equation}
Note that  the  level increases by $1$ in the scope of $\oc$, and 
of any operator $\op\in \OpSet$.

 A reduction step $\Tm  \ToRule  \TmTwo$  
 is \emph{at level} $k$ 
 if it fires a $\Rule$-redex at level $k$; it is 
 \emph{least-level } 
 if it reduces a  redex whose level is minimal.

 The \emph{least level} $\llBang{\Tm}$ of a term $\Tm$ expresses the minimal level of any  occurrence of redexes 
 in $\Tm$;  if no redex is in $\Tm$, we set $\llBang{\Tm} = \infty$. 
%
Formally
\begin{definition}[Least-level reduction] \label{def:ll}
	\label{def:depth-bang}
	Let   $\red\eq   \bigcup_\Rule \redx{\Rule}$ ($ \Rule\in \RulesSet$) and   $\R = \bigcup_\Rule \R_{\Rule}$  the   set of   redexes.  
 Given a function $\lev{-}$ from contexts into   $\Nat $:
\begin{itemize}
	
	\item The  \emph{least level} 
of a term $\Tm$ is  defined as 
\begin{align} 
\label{eq:least-level}
\llb{}{\Tm} \Defeq
\inf \{\lev{\Ctx} \mid  \Tm = \Ctxp{\TmThree} \text{ for some }  \TmThree\in \R \} \in (\Nat \cup \{\infty\}).\footnotemark
\end{align}
\footnotetext{Recall  that $\inf \emptyset = \infty$, when $\emptyset$ is seen as the empty subset of $\Nat$ with the usual order.}

 \item A $\Rule$-reduction step $\Tm  \ToRule  \TmTwo$  is:
\begin{enumerate}
	\item \emph{at level $k$}, written $\Tm \ToRuleAt{k} \TmTwo$, if  $\Tm \Defeq \Ctxp{\TmThree} \ToRule \Ctxp{\TmThree'} \Eqdef \TmTwo$ and $\lev{\Ctx} = k$. 
	\item \emph{least-level},  written $\Tm \llredRule \TmTwo$, if $\Tm \ToRuleAt{k} \TmTwo$ and  $k=\llb{}{\Tm}$.
	\item   \emph{internal}, written $\Tm \nllredRule \TmTwo$, if  $\Tm \ToRuleAt{k} \TmTwo$ and  $k>\llb{}{\Tm}$.
\end{enumerate}

\item \emph{Least-level reduction}  is $\llred \eq \bigcup_\Rule \llredRule$  ($\Rule\in \RulesSet$).	
\item \emph{Internal reduction}  is $\nllred \eq \bigcup_\Rule \nllredRule $   ($\Rule\in \RulesSet$).

\end{itemize}
\end{definition}

Note that $\red \ = \ \llred \!\cup \nllred $.  
Note also that the definition of least level of a term  depends on the set  $\R = \bigcup_\Rule \R_{\Rule}$ of redexes associated with $\red$.\footnotemark
\footnotetext{
	We should write  $ \llBangR{\Tm} $, $\ll_{\R}$ and $\xredx{\ll_{\R}}{\Rule}$, but we avoid it for the sake of readability.}

\paragraph{Normal Forms.}
It is immediate that $\llred\subset \red $ is a \emph{strategy} for $\red$. Indeed, $\llred$ and $\red$ have the \emph{same normal forms} because if $M$ has a $\red$-redex, it has a redex at least-level, \ie it has a $\llred$-redex.

 \begin{remark}[Least level of normal forms] 
 	\label{rmk:least-level-normal}
	Note that  $\llb{}{\Tm} = \infty$ if and only if $\Tm$ is $\red$-normal, because $\lev{\Ctx} \in \Nat$ for all contexts $\Ctx$.
\end{remark}

\paragraph{A \good least-level reduction.}


The beauty of least-level reduction for the bang calculus, is that  it satisfies some elegant properties, which allow for neat proofs, in particular 
monotonicity and internal invariance (in \Cref{def:good}).
The developments  in the rest of the paper rely on such properties,  and in fact will apply to any  calculus whose reduction $\red$ has the properties described below.

\begin{definition}[Good least-level]\label{def:good} A reduction $\red$ has a \emph{\good least-level} if:
	\begin{enumerate}
		\item\label{p:ll-properties-monotone}
		\emph{Monotonocity:}  $\Tm \red \TmTwo$ implies $\llBang{\Tm} \leq \llBang{\TmTwo}$.
		
		\item\label{p:ll-properties-invariance}
		\emph{Internal invariance:}  $\Tm \nllred \TmTwo$ implies $\llBang \Tm = \llBang{\TmTwo}$. 
			\end{enumerate}
		
\end{definition}
%
%
%
%
%
%
%
Point 1.  states that no step can decrease the least level of a term. Point 2.  says that internal steps cannot change the least level of a term. 
Therefore, only least-level steps may increase the least level.
Together, they  imply  persistence:  only least-level steps can approach normal forms.
%
%
\begin{fact}[Persistence] 
	\label{p:ll-properties-persistence}
	If $\red$ has a good least-level, then $\Tm \nllred \TmTwo$ 
		implies $\TmTwo$ is not $\red$-normal.
	
\end{fact}



The  pure bang calculus $({\BangSet}, \redbb) $  has  a \good least-level; the same holds true when extending the reduction with operators.
\begin{restatable}[Good least-level of  bang calculi]{proposition}{llproperties}\label{prop:ll-properties}	
	Given ${\BangSetOp}$, let 
	  $\red  \eq \redbb\cup \redx{\OpSet}$, where each  $\op\in \OpSet$ has a redex of shape $\op(P_1, \dots, P_k)$. The reduction  $\red$ has  a \good least-level.
\end{restatable}

\subsection{Least-level for a bang calculus: examples.}
Let us examine more closely least-level reduction for a  bang calculus $(\BangSetOp, \To{})$.
%
For concreteness, we  consider  $\RulesSet =\{\bbeta, \op\mid \op\in \OpSet\}$, hence the set of redexes is  $\R=\R_{\betab} \cup \R_\OpSet$, where $\R_\OpSet$ is a set of terms of shape $\op(P_1,\dots, P_k)$.

%
We observe  that the least level $\llb{}{\Tm}$ of a term $\Tm\in \BangSetOp$ 	can easily be defined in a direct way,  inductively.
\begin{itemize}
	\item  $\llb{}{\Tm} = 0 \text{ if }\Tm \in \R\eq\R_{\betab} \cup \R_\OpSet$, 
	\item 	otherwise:\\
	$ 	
	\llb{}{\var} = \infty \quad 
	\llb{}{\La{\var}\Tm} = \llb{}{\Tm}
	\quad \llb{}{\Bang{\Tm}} = \llb{}{\Tm} + 1 \quad
	\llb{}{\Tm \TmTwo} = \min\{\llb{}{\Tm} , \llb{}{\TmTwo} \} $
	
\end{itemize}

\begin{example}[Least level of a term]\label{ex:ll}
  Let $R \in \R_{\bbeta} $.    	If $T_0:=R\,(!R$), then $\llb{}{T_0}= 0$. 
	If  $T_1:= x!R$ then $\llb{}{T_2}= 1$.
	If  $T_2 := \op(x,y)!R$ then $\llb{}{T_2}=0$, as $\op(x,y) \in \R_{\op}$.

\end{example}

Intuitively, least-level reduction fires a redex that is \emph{minimally nested}, where a redex is any subterm whose form is in $\R \eq\R_{\betab} \cup \R_\OpSet$.
 Note  that least-level reduction can choose to fire one among possibly \emph{several} redexes at minimal level.
 
\begin{example}\label{ex:llred1}Let us revisit \Cref{ex:ll} with $R= (\lam x.x) \Bang{z}\in \R_{\bbeta} $ ($R\Root{\bbeta} z$). Then 
	$T_1:= x\,\Bang{R} \llredbb x\,\Bang{z}$ but   $T_0:= R\,(\Bang{R}) \not\llred R\, \Bang{z}$ and similarly
	$T_2 := \op(x,y)\,\Bang{R} \not\llredbb \op(x,y) !z$.
Observe also that $\op(x,\underline{R}) \not\llredbb \op(x,z ) $.
\end{example}

\begin{example}\label{ex:llred2}  Let $R = (\lam x.x) \Bang{z}$.   Two least-level steps are possible in ${(\lam z. R)\Bang{R} }$: 
	$\underline{(\lam z. R)\Bang{R} }\llredbb (\lam x.x)\Bang{R}$, and  $(\lam z. \underline{R})\Bang{R} \llredbb (\lam z. z)\Bang{R}$. 
	But $(\lam z. R) \Bang{\underline{R}} \not\llredbb (\lam z. R)!z$.

\end{example}

\SLV{}{Finally, let us revisit Examples \ref{ex:llred1} and \ref{ex:llred2} making explicit the level.

\begin{example*}[\ref{ex:llred1}, revisited]
	$T_1:= x!(\id !z) \redx{\bbeta:1} x!z$ and $\llb{}{T_1}=1$.
	$T_2 := \op(x,y)!(\id !z) \redx{\bbeta:1} \op(x,y) !z$ and $\llb{}{T_2}=0$.
	Note  that $\op(x,\underline{R}) \redx{\bbeta:1} \op(x,z ) $, while $\llb{}{\op(x,\underline{R}) }=0$.
\end{example*}

\begin{example*}[\ref{ex:llred2}, revisited] Let $R,S$ as in \Cref{ex:llred2}. Then 
	$\underline{(\lam z. R)!S }\redx{\bbeta:0} \id!S$,   $(\lam z. \underline{R})!S \redx{\bbeta:0} (\lam x. z)!S$, and 
	$(\lam z. R)!\underline{S} \redx{\bbeta:1} (\lam z. R)!z$, with $\llb{}{(\lam x. R)!S} =0$.

\end{example*}

}

\subsection{Least-level for \CbN and \CbV $\lambda$-calculi}
\label{subsect:level-cbn-cbv}
The definition of least-level reduction in \Cref{sec:ll_def} is independent from the specific notion of level that is chosen, and also from the specific calculus. 
The idea is that the reduction strategy  persistently   fires a redex at minimal level, once such a notion is set.

Least-level reduction can indeed be defined also for the  \CbN and \CbV $\lambda$-calculi, given  an opportune definition of level. 
In \CbN, we count the number of nested arguments and operators containing the occurrence of redex.
In \CbV, we count the number of nested operators and \emph{unapplied} abstractions containing the redex, where an abstraction is unapplied if it is not the right-hand side of an application.
Formally, an occurrence of redex  is identified by a context (as explained in \Cref{sec:ll_def}), and we define the following $ \levCbn{\cdot} $ and $ \levCbv{\cdot} $ functions from $\LambdaCtx$ to $\Nat$, the \emph{level} in \CbN and \CbV $\lam$-calculi.
 {\footnotesize \begin{align*}
	\label{eq:level-cbn-cbv}
	\levCbn{\Hole{\cdot}} &= 0 
	& 
	\levCbv{\Hole{\cdot}} &= 0
	\\
	\levCbn {\La{\var}\LCtx}&= \levCbn {\LCtx} 
	& 
	\levCbv {\La{\var}\LCtx}&= \levCbv {\LCtx} + 1 
	\\
	\levCbn {\LCtx \LTm} &= \levCbn {\LCtx} 
	& 
	\levCbv {\LCtx \LTm} &= 
	\begin{cases}
		\levCbv{\LCtxTwo} &\text{if } \LCtx = \La{\var}{\LCtxTwo} \\
		\levCbv{\LCtx} &\text{otherwise}
	\end{cases}
	\\
	\levCbn {\LTm \LCtx} &= \levCbn {\LCtx} + 1 
	& 
	\levCbv {\LTm \LCtx} &= \levCbv {\LCtx}  
	\\
	\levCbn {\op (\dots, \LCtx, \dots)}&= \levCbn {\LCtx} +1 
	& 
	\levCbv {\op (\dots, \LCtx, \dots)}&= \levCbv {\LCtx} +1
\end{align*}}
In both \CbN and \CbV $\lambda$-calculi, the \emph{least level} of a term (denoted by $\llCbn{\cdot}$ and $\llCbv{\cdot}$) and \emph{least-level} and \emph{internal} reductions are given by \Cref{def:depth-bang} (replace $\lev{\cdot}$  with $\levCbn{\cdot}$ for \CbN and $\levCbv{\cdot}$ for \CbV).

In \Cref{sect:embedding} we will see that the definitions of \CbN and \CbV least level are not arbitrary, but induced by the \CbN and \CbV translations defined in \Cref{subsect:translations}.

\section{Embedding of \CbN and \CbV by level}
\label{sect:embedding}

Here we refine the analysis of the \CbN and \CbV translations given in \Cref{subsect:translations}, by showing two new results: translations preserve normal forms (\Cref{prop:preservation-normal}) and least-level (\Cref{prop:preservation-reduction}), back and forth. 
This way, to obtain  least-level  \emph{factorization} or least-level \emph{normalization} results, it suffices to prove them   in the bang calculus. The translation transfers the results into the \CbN and \CbV $\lambda$-calculi  (\Cref{thm:translation}).
We use here the expression ``translate'' in a strong sense: the results for \CbN and \CbV $\lambda$-calculi are obtained from the corresponding results in the bang calculus almost for free, just via \CbN and \CbV translations.

\paragraph{Preservation of normal forms.}
The targets of the \CbN translation $\Cbn{(\cdot)}$  and \CbV translation $\Cbv{(\cdot)}$
into the bang calculus can be \emph{characterized syntactically}. 
A fine analysis of these fragments of the bang calculus (see 
\cite{bangLong} for details)
proves that both \CbN and \CbV translations
preserve normal forms, back and forth.

\begin{proposition}[Preservation of normal forms]
	\label{prop:preservation-normal}
	Let $\LTm, \LTmTwo \in \Lambda$ and $\op \in \OpSet$. 
	\begin{enumerate}
		\item\label{p:preservation-normal-cbn} \emph{\CbN:} $\LTm$ is $\beta$-normal 
		iff $\Cbn{\LTm}$ is $\Tot$-normal; $\LTm$ is $\op$-normal 
		iff $\Cbn{\LTm}$ is $\op$-normal.
		\item\label{p:preservation-normal-cbv} \emph{\CbV:} $\LTm$ is $\Betav$-normal 
		iff $\Cbv{\LTm}$ is $\Tot$-normal; $\LTm$ is $\op$-normal 
		iff $\Cbv{\LTm}$ is $\op$-normal.
	\end{enumerate}
\end{proposition}

By \Cref{rmk:least-level-normal}, \Cref{prop:preservation-normal} can be seen as the fact that \CbN and \CbV translations preserve the least-level of a term, back and forth, when the least-level is infinite.
Actually, this holds more in general for any value of the least-level.


\paragraph{Preservation of levels.} 
We aim to show that least-level steps in \CbN and \CbV $\lambda$-calculi correspond to least-level steps in the bang calculus, back and forth, via \CbN and \CbV translations respectively (\Cref{prop:preservation-reduction}).
This result is subtle, one of the main technical contributions of this paper. 

First, we extend the definition of translations to contexts.
The \emph{\CbN and \CbV translations for contexts} are two functions $\Cbn{(\cdot)} \colon \LambdaCtx \to \BangSetCtx$ and $\Cbv{(\cdot)} \colon \LambdaCtx \to \BangSetCtx$, respectively, mapping contexts of the $\lambda$-calculus into contexts of the bang calculus: 

{\footnotesize \vspace{-2\baselineskip}
\begin{center}
	\small
	\begin{align*}
	\Cbn{\Hole{\cdot}} &= \Hole{\cdot}				
	& 
	\Cbv{\Hole{\cdot}} &= \Hole{\cdot} 
	\\
	\Cbn{(\La{\var}{\LCtx})} &= \La{\var}{\Cbn{\LCtx}} 
	&
	\Cbv{(\La{\var}{\LCtx})} &= \Bang{(\La{\var}{\Cbv{\LCtx}})} 
	\\
	\Cbn{(\op(\LTm_1, ... , \LCtx, ... , \LTm_k))} &= \op(\Cbn{\LTm_1}, ... , \Cbn{\LCtx}, ... , \Cbn{\LTm_k})
	&\quad
	\Cbv{(\op(\LTm_1, ..., \LCtx, ..., \LTm_k))} &= \op(\Cbv{\LTm_1}, ..., \Cbv{\LCtx}, ..., \Cbv{\LTm_k})
	\\
	\Cbn{(\LCtx\LTm)} &= \Cbn{\LCtx}\,{\Bangp{{\Cbn{\LTm}}}} 
	& 
	\Cbv{(\LCtx\LTm)} &= 
	\begin{cases}
	\Ctx \, \Cbv{\LTm} &\!\!\text{if } \Cbv{\LCtx} = \Bang{\Ctx}
	\\
	(\Der{\Cbv{\LCtx}}){\Cbv{\LTm}} &\!\!\text{otherwise}
	\end{cases} 
	\\
	\Cbn{(\LTm\LCtx)} &= \Cbn{\LTm}\,{\Bangp{{\Cbn{\LCtx}}}} \,; 
	& 
	\Cbv{(\LTm\LCtx)} &= 
	\begin{cases}
	\Tm \, \Cbv{\LCtx} &\!\!\text{if } \Cbv{\LTm} = \Bang{\Tm}
	\\
	(\Der{\Cbv{\LTm}}){\Cbv{\LCtx}} &\!\!\text{otherwise.}
	\end{cases}
	\end{align*}
\end{center}}

Note that \CbN (\Resp \CbV) level of a context defined in \Cref{subsect:level-cbn-cbv} increases by $1$ whenever the \CbN (\Resp \CbV) translation for contexts add $\oc$.
Thus, \CbN and \CbV translations preserve, back and forth, 
the level of a redex and the least-level of a term.
Said differently, the level for \CbN and \CbV is defined in \Cref{subsect:level-cbn-cbv} so as to enable the preservation of level via \CbN and \CbV translations. 

\begin{lemma}[Preservation of level via \CbN translation]
	\label{lemma:preservation-level-cbn}
	\hfill
	\begin{enumerate}
		\item\label{p:preservation-level-cbn-context}  \emph{For contexts:} For any context $\LCtx \in \CtxSet$, one has $\levCbn{\LCtx} = \lev{\Cbn{\LCtx}}$.
		
		\item\label{p:preservation-level-cbn-reduction} \emph{For reduction:} For any term $\LTm \in \LambdaOp$: $\LTm \ToBetaInd{k} \LTmTwo$ if and only if $\Cbn{\LTm} \ToTotInd{k} \Cbn{\LTmTwo}$;
		and $\LTm \ToOpAt{k} \LTmTwo$ if and only if $\Cbn{\LTm} \ToOpAt{k} \Cbn{\LTmTwo}$, for any $\op \in \OpSet$.
		
		\item\label{p:preservation-level-cbn-ll} \emph{For least-level of a term:} For any term $\LTm \in \LambdaOp$, one has $\llCbn{\LTm} = \llBang{\Cbn{\LTm}}$.
	\end{enumerate}
\end{lemma}

\begin{lemma}[Preservation of level via \CbV translation]
	\label{lemma:preservation-level-cbv}
	\hfill
	\begin{enumerate}
		\item\label{p:preservation-level-cbv-context}  \emph{For contexts:} For any context $\LCtx \in \CtxSet$, one has $\levCbv{\LCtx} = \lev{\Cbv{\LCtx}}$.
		
		\item\label{p:preservation-level-cbv-reduction} \emph{For reduction:} For any term $\LTm \in \LambdaOp$: $\LTm \ToBetavAt{k} \LTmTwo$ if and only if $\Cbv{\LTm} \ToBangAt{k} \ToDerInd{k}^= \Cbv{\LTmTwo}$; and $\LTm \ToOpAt{k} \LTmTwo$ if and only if $\Cbv{\LTm} \ToOpAt{k} \ToDerInd{k}^= \Cbv{\LTmTwo}$, for any $\op \in \OpSet$. 
		
		\item\label{p:preservation-level-cbv-ll} \emph{For least-level of a term:} For any term $\LTm \in \LambdaOp$, one has $\llCbv{\LTm} = \llBang{\Cbv{\LTm}}$. 
	\end{enumerate}
\end{lemma}	

From the two lemmas above it follows that \CbN and \CbV translations preserve least-level and internal reductions, back and forth.

\begin{proposition}[Preservation of least-level and internal reductions]
	\label{prop:preservation-reduction}
	Let $\LTm$ be a \lam-term and $\op \in \OpSet$.
	\begin{enumerate}
		\item\label{p:preservation-reduction-ll-cbn} \emph{\CbN least-level:} $\LTm \llredb \LTmTwo$ iff $\Cbn{\LTm} \llredBang \Cbn{\LTmTwo}$; and $\LTm \llredOp \LTmTwo$ iff $\Cbn{\LTm} \llredOp \Cbn{\LTmTwo}$.
		\item\label{p:preservation-reduction-nll-cbn} \emph{\CbN internal:} $\LTm \nllredb \LTmTwo$ iff $\Cbn{\LTm} \nllredBang \Cbn{\LTmTwo}$; and $\LTm \nllredOp \LTmTwo$ iff $\Cbn{\LTm} \nllredOp \Cbn{\LTmTwo}$.
		\item\label{p:preservation-reduction-ll-cbv} \emph{\CbV least-level:} $\LTm \llredbv \LTmTwo$ iff $\Cbv{\LTm} \llredBang\llredDer^= \Cbv{\LTmTwo}$; and $\LTm \llredOp \LTmTwo$ iff $\Cbv{\LTm} \llredOp\llredDer^= \Cbv{\LTmTwo}$.
		\item\label{p:preservation-reduction-nll-cbv} \emph{\CbV internal:} $\LTm \llredbv \LTmTwo$ iff $\Cbv{\LTm} \nllredBang \nllredDer^= \Cbv{\LTmTwo}$; and $\LTm \llredOp \LTmTwo$ iff $\Cbv{\LTm} \nllredOp \nllredDer^= \Cbv{\LTmTwo}$.
	\end{enumerate}
\end{proposition}

As a consequence,  least-level reduction induces factorization in \CbN and \CbV $\lambda$-calculi as soon as it does in the bang calculus.
And, by \Cref{prop:preservation-normal}, it is a normalizing strategy in \CbN and \CbV 
as soon as it is so in the bang calculus.

\begin{theorem}[Factorization and normalization by translation]\label{thm:translation} 
Let 
$\LambdaOp^\cbn\eq (\LambdaOp, \redb \cup \redx{\OpSet})$ and  $\LambdaOp^\cbv\eq (\LambdaOp,\ \redbv \cup \redx{\OpSet})$.
	\begin{enumerate}
		\item If $\BangSetOp$ admits  least-level factorization $\F{\llred}{\nllred}$, then so do  
		$\LambdaOp^\cbn$ and  $\LambdaOp^\cbv$.
		\item If $\BangSetOp$ admits  least-level normalization, then so do $\LambdaOp^\cbn$ and $\LambdaOp^\cbv$.
\end{enumerate}
\end{theorem}

A similar result will hold also when extending the pure calculi with a rule $\Root\Rule$ other than $\Root{\op}$, as long as the translation preserves redexes.

\begin{remark}[Preservation of least-level and of normal forms.]
Preservation of normal form and least-level is delicate.
For instance, it does not hold with the definition \CbV translation $\Cbv{(\cdot)}$ in \cite{GuerrieriManzonetto18,SantoPintoUustalu19}. There, the translation $\tm = \tmThree \tmTwo \in \Lambda$ would be $\Cbv{\tm} = (\Der{\Bangp{\Cbv{\tmThree}}}) \Cbv{\tmTwo}$ and then \Cref{prop:preservation-normal} and \Cref{prop:preservation-reduction} would not hold: $\Der{\Bangp{\Cbv{\tmThree}}}$ is a $\Tot$-redex in $\Cbv{\tm}$ (see \Cref{ex:identity}) and hence $\Cbv{\tm}$ would not be normal  even though so is $\tm$, and $\llBang{\Cbv{\tm}} = 0$ even though $\llCbv{\tm} \neq 0$.
This is why we defined two distinct case when defining $\Cbv{(\cdot)}$ for applications, akin to \cite{BucciarelliKesnerRiosViso20}. 
\end{remark}
\section{Least-level factorization via   bang calculus}
We have shown that  least-level factorization in a bang calculus $\BangSetOp$ implies least-level factorization in the corresponding \CbN and \CbV calculi,  via forth-and-back translation. 
The central question now is \emph{how to prove least-level factorization} for a bang calculus: the rest of the paper is devoted to that.

\subsubsection{Overview.}
\label{sec:overview}
 Let us overview our approach by considering $\OpSet=\{\op\}$, and $\red \eq \redbb \cup \redx{\op}$. 
Since by definition $\llred = \llredbb \cup \llredx{\op}$ (and $\nllred = \nllredbb \cup \nllredx{\op}$), 
 \Cref{thm:modular} states that we can \emph{decompose} least-level factorization of $\red$ in three modules:
 \begin{enumerate}
 	\item prove \ll-factorization of $\redbb$, \ie ~~
 	$\redbb^* ~\subseteq~ \llredbb^* \cdot \nllredbb$
 	 \item prove \ll-factorization of $\redx{\op}$, \ie ~~  	$\redx{\op}^* ~\subseteq~ \llredx{\op}^* \cdot \nllredx{\op}$
 	 \item prove the two linear swaps of \Cref{thm:modular}.
 \end{enumerate}
Please note that the least level for both $\llredbb$ and $\llredx{\op}$  is   defined  with respect to the  redexes  $\R=\R_{\bbeta} \cup \R_{\op}$, so  to have $\llred = \llredbb \cup \llredx{\op}$. This  addresses  the issue we mentioned in \Cref{ex:issue}.

Clearly, points 2. and 3. depend on the specific rule $\Root{\op}$. However, the beauty of a modular approach is that  point 1. can be established in general: we do not need to know $\Root{\op}$, only the shape of its redexes $\R_{\op}$. 
In  \Cref{sect:factorization} we  provide a general result of   \ll-factorization for  $\redbb$ (\Cref{thm:factorize-bang}). In fact, 
we shall show a bit more:   the   way of  decomposing the study of  factorization that we have sketched, can be applied to study least-level factorization of any reduction $\red \eq \redbb \cup \redc$, as long as  $\red$ has a \good least-level. 
	
Once (1.) is established 
(once and for all), to prove factorization of a reduction $\redbb\cup \redx{\op}$ we are only left with (2.) and (3.).  
In \Cref{sec:ll_modular}	 we show that  the proof of the two linear swaps can be reduced to a single, simple test, involving  only the $\Root{\op}$ step  (\Cref{prop:test_ll}).	
In \Cref{sec:case_study}, we will illustrate how  all elements play together on a concrete case, applying them to non-deterministic $\lam$-calculi.
	
%

\subsection{Factorization of $\redbb$ in a bang calculus}
\label{sect:factorization}

We  prove that $\redbb$-reduction \emph{factorizes} via least-level reduction (\Cref{thm:factorize-bang}).
%
The result holds for a definition of  $\llredbb$ (as in  \Cref{sec:ll}) where the set of redexes $\R$ is $\R_{\bbeta} \cup \R_{\op}$---this generalization has essentially no cost, and allows us to use \Cref{thm:factorize-bang} as a module in the factorization of a larger reduction. 

We prove factorization via  Takahashi's  Parallel Reduction method \cite{Takahashi95}. 
We  define a reflexive reduction $\ParallelNotLlBang$ (called parallel internal $\Tot$-reduction) which satisfies the conditions of \Cref{fact:fact}, \ie $ \ParallelNotLlBang^* \eq   \nllredBang^*$ and $\ParallelNotLlBang \!\cdot\! \llredBang \subseteq \llredBang^* \!\cdot\! \ParallelNotLlBang$.

	The tricky point is to prove 
	$\ParallelNotLlBang \!\cdot\! \llredBang \subseteq \llredBang^* \!\cdot\! \ParallelNotLlBang$
	We adapt the proof technique   in \cite{AccattoliFaggianGuerrieri19}. All details are \SLV{in \cite{bangLong}.}{in the appendix.}
	Here we just give the definition of $\ParallelNotLlBang$.

We first introduce 
$\ParallelBangAt{n}$ (the parallel version of $\ToBangAt{n}$), which fires simultaneously a number of $\Tot$-redexes at level at least $n$ 
(and $\ParallelBangAt{\infty}$ does not reduce any $\Tot$-redex: $\Tm \ParallelBangAt{\infty} \TmTwo$ implies $\Tm = \TmTwo$).
\begin{center}
	\small
	\begin{prooftree}
		\infer0{\var \ParallelBangAt{\infty} \var}
	\end{prooftree}
	\quad
	\begin{prooftree}
		\hypo{\Tm \ParallelBangAt{n} \Tm'}
		\infer1{\La{\var}{\Tm} \ParallelBangAt{n} \La{\var}{\Tm'}}
	\end{prooftree}
	\quad
	\begin{prooftree}[separation=1.2em]
		\hypo{\Tm \ParallelBangAt{m} \Tm'}
		\hypo{\TmTwo \ParallelBangAt{n} \TmTwo'}
		\infer2{\Tm\TmTwo \ParallelBangAt{\min\{m,n\}} \Tm'\TmTwo'}
	\end{prooftree}
	\quad
	\begin{prooftree}
		\hypo{\Tm \ParallelBangAt{n} \Tm'}
		\infer1{\Bang{\Tm} \ParallelBangAt{n\!+\!1} \Bang{\Tm'}}
	\end{prooftree}
	\quad
	\begin{prooftree}[separation=1.2em]
		\hypo{\Tm \ParallelBangAt{n} \Tm'}
		\hypo{\TmTwo \ParallelBangAt{m} \TmTwo'}
		\infer2{(\La{\var}{\Tm})\Bang{\TmTwo} \ParallelBangAt{0} \Tm'\Sub{\TmTwo'\!}{\var}}
	\end{prooftree}
\end{center}

The \emph{parallel internal $\Tot$-reduction} $\ParallelNotLlBang$ is the parallel version of $\nllredBang$, which fires simultaneously a number of $\Tot$-redexes that are not at minimal level.
Formally, 
\begin{center}
	\small
	$\Tm \ParallelNotLlBang \TmTwo$ \quad if $\Tm \ParallelBangAt{n} \TmTwo$ with $n = \infty$ or $n > \llBang{\Tm}$.
\end{center}


\begin{restatable}[Least-level factorization of $\redbb$]{theorem}{factorizebang}
	\label{thm:factorize-bang} 
	Assume 
	$\red =\redbb \cup \redx{\Rule}$ has \good least-level in $\BangSetOp$.  
	Then:
	$	 \Tm \ToTot^* \TmTwo \text{ implies } \Tm \llredBang^* \!\cdot\! \nllredBang^* \TmTwo.$
\end{restatable}

\begin{corollary}[Least-level factorization in  the pure bang calculus]	\label{cor:factorize-bang}
	In the pure bang calculus $(\BangSet, \ToTot)$,  if $\Tm \ToTot^* \TmTwo$ then $\Tm \llredBang^* \!\cdot\! \nllredBang^* \TmTwo$.

		
\end{corollary}

\SLV{}{
\subsubsection{Surface Reduction.}
\emph{Surface reduction} $\sredbb$ (defined by Simpson in \cite{Simpson05}) is the reduction which only  reduces  a redex  at level $0$ (any such redex). 
It fires redexes that are not inside boxes.
It can be equivalently defined as the closure of $\Root{\bbeta}$ under contexts  $\WCtx$ defined by the grammar $\WCtx  	\Coloneqq \, \Hole{\cdot} \mid \La{\var}{\WCtx} \mid \WCtx\Tm \mid \Tm\WCtx  $.
We write 
$\Tm \sred\TmTwo$ for a step $\red$ which is surface; 
otherwise, we write $\Tm \nsred \TmTwo$. Clearly $\sred\subset \llred$.

\emph{Surface factorization} was already proven in \cite{Simpson05}:  
\begin{center}
	$\Tm \ToTot^* \TmTwo$ implies $\Tm \sredBang^* \!\cdot\! \nsredBang^* \TmTwo$.
\end{center}
We obtain this result as 
a consequence of least-level factorization (\Cref{cor:factorize-bang} and monotonicity (\Cref{prop:ll-properties}) of least-level reduction.
}

%
%
%


\subsection{Pure calculi and least-level normalization}
Least-level factorization of $\redbb$ implies in particular least-level factorization for $\redb$ and $\redbv$. As a consequence, 
least-level reduction is a normalizing strategy for all three pure calculi: the bang calculus, the \CbN, and the \CbV $\lam$-calculi. 
\subsubsection{The pure bang calculus.}\label{sec:bang_strategy}
 $\llredbb$  is a \textit{normalizing strategy} for $\redbb$. Indeed,  it satisfies all ingredients in \Cref{prop:abs_normalization}. 
%
%
%
%
Since we have least-level factorization (\Cref{cor:factorize-bang})), same normal forms, and  \emph{persistence}   (\Cref{prop:ll-properties}), 
$\llredbb$ is a \emph{complete strategy} for $\redbb$:
$\text{If } N \text{ is }\bbeta\text{-normal and }  M \redbb^* N,    \text{ then }   M \llredbb^* N.$

We already observed (\Cref{ex:llred2}) that the least-level reduction $\llredbb$ is non-deterministic, because several redexes at least level may be available. 
Such non-determinism is however inessential, because $\llredbb$ 
is \emph{uniformly normalizing}.

\begin{lemma}[Quasi-Diamond]\label{lemma:diamond}
	In  $(\BangSet, \ToTot)$, the reduction $\llredbb $ is quasi-diamond (\Cref{fact:diamond}), and therefore uniformly normalizing.
	
	
\end{lemma}

Putting all the ingredients together, we have (by \Cref{prop:abs_normalization}):


\begin{restatable}[Least-level normalization]{theorem}{normalizebang}
	\label{thm:normalize-bang}
	In the pure bang calculus $\llredbb$ is a normalizing strategy for $\redbb$.
\end{restatable}
 \Cref{thm:normalize-bang} means not only that if $\Tm$ is $\Tot$-normalizable then $\Tm$ can reach its normal form by just performing least-level steps, but also that performing \emph{whatever} least-level steps eventually leads to the normal form, if any.

\subsubsection{Pure \CbV and \CbN $\lam$-calculi.}
	By forth-and-back translation (\Cref{thm:translation}) the  least-level factorization and  normalization results for the pure bang calculus  immediately transfers to the \CbN and \CbV setting.
	
	\begin{theorem}[\CbV and \CbN least-level normalization]
		\begin{itemize}
			\item \CbN: In  $(\Lambda, \redb)$, $\llredb$ is a normalizing strategy for $\redb$.
			\item \CbV: In  $(\Lambda, \redbv)$, $\llredbv$ is a normalizing strategy for $\redbv$.
		\end{itemize}
	\end{theorem}

\subsection{Least-level Factorization, Modularly. }\label{sec:ll_modular}

As anticipated at the beginning of this section, we can use  \Cref{thm:factorize-bang} also as part of the proof of factorization for  a more complex calculus.
We now introduce one more useful  tool: a simple test to establish least-level factorization of a reduction $\redbb \cup \redx{\Rule} $ (where $\redx{\Rule}$ is a new reduction added to $\redbb$). 
We shall give an example of its use  in \Cref{sec:case_study} (see the proof of \Cref{thm:ND_fact}).

The test embodies  \Cref{thm:modular}, and the fact that we already know (once for all) that $\redbb$ factorizes via $\llredbb$. 
It   turns out that the proof of the two linear swaps can be reduced to a single, simple test, which  only involves the $\Root{\Rule}$
step.

\begin{restatable}[Test for {modular} least-level factorization]{proposition}{testll}\label{prop:test_ll} 
Let  $\redc$ be   the contextual closure of a rule $\Root{\Rule}$, and 
assume 	
	  $\red \eq( \redbb \cup \redc)$ has a \good least-level.  Then 	$\red$  factorizes via $\llred \eq (\llredbb \cup \llredc)$ if the following hold:
	\begin{enumerate}
		\item \emph{$\ll$-factorization of $\redc$}: ~~
		$\redc^* ~\subseteq ~\llredc^* \cdot \nllredc^*$
		\item  $\rredc$ is  \emph{substitutive}:  ~~ 	$R \rredc R' 
		\text{ implies } R \subs x Q \rredc R'\subs x Q.$
		\item \emph{Root linear swap}: ~~ $\nllredbb \cdot \rredc  \ \subseteq \ \rredc \cdot\redbb^* $.
	\end{enumerate}
\end{restatable}
Note that, as usual, at point (1.) 	 the least level is defined w.r.t. $\R=\R_{!\beta}\cup \R_{\rho}$.

\SLV{}{

\paragraph{A modular test for \CbN and \CbV least-level factorization.} A modular test analogous to \Cref{prop:test_ll} holds also in \CbN and \CbV: simply replace in the statement  $\bbeta$ with $\beta$ (for \CbN) or 
$\betav$ (for \CbV).
}

\section{Case study:  non-deterministic $\lam$-calculi}\label{sec:case_study}
\newcommand{\Bango}{{\BangSet}_{\oplus}}

To 
show how to use our framework, we apply the set of tools which we have developed  on our running example.
We extend the bang calculus with a non-deterministic operator, then considering $(\Bango, \red_{\bbeta\oplus})$ where $\red_{\bbeta\oplus}  \eq (\redbb\cup\redo)$,  
and $\redo$ is the contextual closure of the  (non-deterministic)  rules:
	\begin{equation}
	\op(P,Q)\Root{\oplus} P   \quad \quad	\op(P,Q) \Root{\oplus} Q
\end{equation}

 \paragraph{First step: non-deterministic bang calculus.} We   analyze  $\Bango$. 
 We use our modular test to prove least-level factorization for $\Bango$:
 if $\red_{\bbeta\oplus}^*U  \text{ then } T \llredx{\Tot\oplus}^* \cdot  \nllredx{\Tot\oplus}^* U$.
	By \Cref{prop:abs_normalization}, an 
immediate consequence of the factorization result  is that the least-level strategy is \emph{complete}, \ie 
if $U$ is  normal:	$T \red_{\bbeta\oplus}^*U$ implies $T \llredx{\beta\oplus}^*U$.
\SLV{}{We can apply \Cref{prop:abs_normalization} for the same reasons as in \Cref{sec:bang_strategy}  : $\red_{\bbeta\oplus}$ and $\llred_{\bbeta\oplus}$ have the same normal forms, and moreover  	 \emph{persistence} holds  (by \Cref{prop:ll-properties}).}

\paragraph{Second step: \CbN and \CbV non-deterministic calculi.} By  translation, we have \textit{for free}, that the analogous results hold in $\PLambda^\cbn$ and $\PLambda^\cbv$,
as defined in \Cref{ex:NDext}.  
So, least-level factorization holds  for both calculi, and moreover 
\begin{itemize}
	\item \emph{\CbN completeness}: in $\PLambda^\cbn$, if $\tmu$ is normal:~ $\tm \red_{\beta\oplus}^*\tmu$ implies $\tm \llredx{\beta\oplus}^*\tmu$.
	\item \emph{\CbV completeness}: in $\PLambda^\cbv$, if $\tmu$ is normal:~  $\tm \red_{\betav\oplus}^*\tmu$ implies  $\tm \llredx{\betav\oplus}^*\tmu$.
\end{itemize}

\paragraph{What do we really need to prove?}
The only result  we need to prove is least-level factorization of $\red_{\bbeta\oplus}$. 
Completeness then follows by \Cref{prop:abs_normalization} and the translations will automatically take care of transferring the results. 

To prove factorization of  $\red_{\bbeta\oplus}$, most of the work  is
done, as $\ll$-factorization of  $\redbb$ is already established; we then  use  our test (\Cref{prop:test_ll}) to extend $\redbb$ with $\redo$.
\SLV{}{
To expose how neat and compact this is, we give explicitly the full proof---it turns out to be just a few lines. }
The only ingredients  we need    are  substitutivity of $\oplus$ (which is an obvious property), and the following easy  lemma.
\begin{lemma}[Roots]\label{l:oplus_basic} Let $\gamma\in \{\bbeta,\oplus	\}$.
	If	$T \nllredc P \mapsto_{\oplus} \TmTwo$ then  $T \mapsto_{\oplus} \cdot  \red_{\gamma} ^= \TmTwo$.
\end{lemma}

\SLV{}{\begin{proof}
Let $P=\oplus (P_1,P_2)$, and consider $P \mapsto_{\oplus} P_1 = \TmTwo$ (the 
case 	$P \mapsto_{\oplus} P_2 = \TmTwo$ is similar). 
Assume $T\nllredx{\Rule} P$. 
Hence $T = \oplus(T_1,T_2)$, and 
either    
$T_1\ToRule P_1$ with $T_2 = P_2$, or $T_2\ToRule P_2$ with $T_1 = P_1$. 
Therefore,     $\oplus( T_1,T_2) \mapsto_{\oplus}   T_1
\ToRule^=  P_1 $.	
\qed
\end{proof}}

\begin{theorem}[Least-level factorization]\label{thm:ND_fact} 
	\begin{enumerate}
		\item\label{thm:ND_fact-bang} In $(\Bango, \to)$, $\F{\llred}{\nllred}$ holds for $\red \eq \redo \cup \redbb$.
		\item Least-level factorization holds in $(\PLambda^\cbn,\redo \cup \redb)$, and in $(\PLambda^\cbv, \redo \cup \redbv)$.
	\end{enumerate} 
\end{theorem}
\begin{proof}
	\begin{enumerate}
		\item It is enough to verify the hypotheses of \Cref{prop:test_ll}.
\SLV{}{	\begin{enumerate}
			\item \emph{Least-level factorization of $\redo$}:~ $\F{\llredo}{\nllredo}$. By   \Cref{l:oplus_basic} (with $\gamma=\oplus$) we obtain $\nllredo \cdot  \llredo   ~\subseteq ~ \llredo \cdot  \redo^= $ \SLV{}{(use \Cref{l:ll_swaps})}, \ie     $\nllredo$ strongly  postpones after $\llredo$. We conclude by \Cref{l:SP}.

			\item \emph{Substitutivity}: ~ {$\mapsto_{\oplus}$ is substitutive.}  Indeed, 
			$$\oplus (P_1,P_2)\subs x Q  = \oplus (P_1 \subs x Q, P_2\subs x Q)~
			\mapsto_{\oplus}~ P_i\subs x Q.$$
			
			\item \emph{Root linear swap}: 
			{$\nllredbb \!\cdot\! \mapsto_{\oplus}\ \subseteq\ \mapsto_{\oplus} \!\cdot\!  \red_{\beta} ^=$.} This is  \Cref{l:oplus_basic} (with $\gamma=\bbeta$).
		\end{enumerate}
	}
		\item It follows from \Cref{thm:translation} and \Cref{thm:ND_fact}.\ref{thm:ND_fact-bang}.
		\qed
	\end{enumerate}

\end{proof}


\SLV{}{
Completeness is the best that can be achieved in these calculi, because of the  true non-determinism of $\redo$ and hence of  least-level reduction and of any other complete strategy for $\to$. 
For instance, in $\PLambda^\cbn$ there is no normalizing strategy for $\oplus(\var, \delta\delta)$ in the sense of \Cref{def:strategy}, since $x  \xrevredx{\ll}{\oplus}\oplus(\var, \delta\delta) \llredx{\oplus} \delta\delta \llredx{\beta}   \dots$\,.


}

\vspace*{-6pt}
\section{Conclusions and Related Work}
The combination of translations (\Cref{thm:translation}), $\ll$-factorization for $\redbb$ (\Cref{thm:factorize-bang}), and modularity (\Cref{prop:test_ll}),
give us a powerful method to analyze  factorization in various $\lam$-calculi that \emph{extend} the pure \CbN and \CbV calculi.
The main novelty is transferring the results from a calculus to another via translations.

We chose to study least-level reduction as a normalizing strategy 
because it is natural to define in the bang calculus, and it is easier to transfer via translations to \CbN and \CbV calculi than leftmost-outermost. 
Since leftmost-outermost is the most common normalizing strategy in \CbN, it is worth 
noticing that  least-level normalization implies leftmost-outermost normalization (and vice-versa).
This is an easy consequence of the---easy to check---fact that 
their union is quasi-diamond (and hence, again, uniformly normalizing). A proof of least-level normalization lends a proof also of leftmost-outermost normalization.

\paragraph{Related Work.}
Many calculi inspired by linear logic subsumes \CbN and \CbV, such as \cite{BentonBPH93,BentonWadler96,RonchiRoversi97,MaraistOderskyTurnerWadler99} (other than the ones already cited).
We chose the bang calculus for its simplicity, which eases the analysis of the \CbN and \CbV translations.

Least-level reduction is studied for linear-logic-based calculi in \cite{Terui07,Accattoli12} and for linear logic proof-nets in \cite{CarvalhoPF11,PaganiTranquilli17}.
Least-level factorization and normalization for the pure \CbN $\lambda$-calculus is studied in \cite{AccattoliFaggianGuerrieri19}.


\bibliographystyle{splncs04}
\bibliography{biblio}

\SLV{}{
\newpage
\appendix
\section*{Technical Appendix}

\section{\Cref{sect:lambda-calculi}: General properties of the contextual closure }

	We start by recalling a basic but key  property of contextual closure. If a step $\redc$ is obtained by closure under \emph{non-empty context} of a rule $\rredc$, then it preserve the shape of the term.
	{We say that $T$ and $T'$ have \emph{the same shape} if  both terms are an 
		application (resp. an abstraction, an atom, a term of shape $!P$ or $\op {(P_1, \dots,P_k)}$ ).}
	\begin{fact}[Shape preservation]\label{fact:shape} 
		Assume $T=\cc\hole{R}\red \cc\hole {R'}=T'$ and that the  context   $\cc$ is \emph{non-empty}. Then $T$ and $T'$ have the same shape.
	\end{fact}

\subsection{Surface Reduction.}Reduction at level $0$ has a key role --  we call reduction which only  fires  a redex  at level $0$ 
	\emph{surface reduction} and write $\sred$. %
Surface reduction	(already defined by Simpson in \cite{Simpson05}); it only 
	 fires redexes that are not inside  a box, or inside the scope of an operator. 
	It can be equivalently defined as the closure of  a rule $\Root{}$ under contexts  $\WCtx$ defined by the grammar $\WCtx  	\Coloneqq \, \Hole{\cdot} \mid \La{\var}{\WCtx} \mid \WCtx\Tm \mid \Tm\WCtx  $.
	We write 
	$\Tm \sred\TmTwo$ for a step $\red$ which is surface; 
	otherwise, we write $\Tm \nsred \TmTwo$. 
	Clearly	
\begin{fact}
 $ \sred \subset \llred$ and $\nllred \subset \nsred$, because a reduction at level $0$ is surely least-level.
\end{fact}
Note that a root  steps   $\Root{}$ is both a  \emph{least level} and a \emph{surface} step.

\subsection{Shape preservation for internal steps}
Let $\red$ be the contextual closure of some  rules. 
  \Cref{fact:shape} implies  that $\nllred$ and $ \nsred$  steps always preserve the shape of terms.

\begin{fact}[Internal Steps]\label{fact:isteps}
	By  \Cref{fact:shape}, $\nsred$  and $\nllred$ preserve the shapes of terms. Moreover, the following hold for $\nsred$, and hence for $\nllred$ (recall that $\nllred \subset \nsred$):
	\begin{enumerate}
		\item There is no $T$ such that $T \nsred x$, for any variable $x$;
		\item $T \nsred  \op(U_1,...,U_k)$ implies $T = \op(T_1,...,T_k)$, and there exists $1\leq j \leq k$ such that  $T_j\red U_j$ ($T_i=U_i$ for $i\not=j$).

		\item   $T\nsred !  U_1$ implies $T = ! T_1$ and $T_1\red U_1$. 
		
		\item $T\nsred \lam x. U_1$ implies $T = \lam x. T_1$ and $T_1\nsred U_1$
		\item $T\nsred U_1U_2$ implies  $T = T_1T_2$, with either (i) $T_1\nsred U_1$ (and $T_2=U_2$), 
		or (ii) $T_2\nsred U_2$ (and $T_1=U_1$).  Moreover, $T_1$ and $U_1$ have the same shape, and so   $T_2$ and $U_2$.
		
	\end{enumerate}
\end{fact}

%
%
%
%

\begin{corollary}\label{cor:bo_redex}Assume $T\nsred S$
\begin{itemize}
	\item $T$ is a $\bbeta$-redex iff $S$ is.
	\item $T$ is an $\op$-redex iff $S$ is.
\end{itemize}
\end{corollary}

\begin{example}If $T\red S$ and $S$ is a $\bbeta$-redex, $T$ does not need to be a $\bbeta$-redex. Consider $((\lam x.x)!(\lam x. P))!z \redbb (\lam x. P)  !z$.
	In the following example, $T$ contains no $\bbeta$-redex:  
$\op(\lam x. P, z)!z  \redx{\op} (\lam x. P)  !z$.
\end{example}

\section{Omitted lemmas and proofs of \Cref{subsect:translations}}
\label{sect:embedding-proofs}

Note that, for any $\LTm \in \Lambda$, $\Cbv{\LTm} \in \BoxSet$ if and only if $\LTm$ is a value (\Ie an abstraction or a variable).
As a consequence, the definition of $\Cbv{(\cdot)}$ in the application case $\LTm \LTmTwo$ depends on whether $\LTm$ is a value or not.

\begin{remark}[\CbV translation is $\Derel$-normal]
	\label{rmk:cbv-val-normal}
	It is immediate to prove 
	that $\Cbv{\LTm}$ is $\Derel$-normal (see \eqref{eq:der} for the 
	definition of $\ToBang$), and that if $\Cbv{\LTm} \ToVal \TmThree \ToDer \TmTwo$ then there is \emph{only} one $\Derel$-redex in $\TmThree$ (absent in $\Cbv{\LTm}$ and in $\TmTwo$) and it has been created by 
	the step $\Cbv{\LTm} \ToVal \TmThree$.
\end{remark}

\begin{lemma}[Substitution] 
	\label{lemma:substitution}
	Let 
	$\LTm, \LTmTwo$ be $\lambda$-terms and $\var$ be a variable.
	\begin{enumerate}
		\item\label{p:substitution-cbn}\emph{For \CbN translation:} One has that $\Cbn{\LTm} \Sub{\Cbn{\LTmTwo}\!}{\var} = \Cbn{(\LTm\Sub{\LTmTwo}{\var})}$.
		\item\label{p:substitution-cbv}\emph{For \CbV translation:} If $\LTmTwo$ is such that $\Cbv{\LTmTwo} = \Bang{\TmTwo}$ for some $\TmTwo \in \BangSet$, then $\Cbv{\LTm} \Sub{\TmTwo}{\var} = \Cbv{(\LTm\Sub{\LTmTwo}{\var})}$.
	\end{enumerate}
\end{lemma}

\begin{proof}   The proofs of both points are by induction on 
	$\LTm \in \Lambda$.
	\begin{enumerate}
		\item 
		\begin{itemize}
			\item \emph{Variable:} If $\LTm$ is a variable then there are two subcases.
			If $\LTm \Defeq \var$ then 
			$\Cbn{\LTm} = {\var}$, so $\Cbn{\LTm} \Sub{\Cbn{\LTmTwo}}{\var} = \Cbn{\LTmTwo} = \Cbn{(\LTm\Sub{\LTmTwo}{\var})}$.
			Otherwise $\LTm \Defeq \varTwo \neq \var$ and then $\Cbn{\LTm} = \varTwo$, hence $\Cbn{\LTm} \Sub{\Cbn{\LTmTwo}}{\var} = \varTwo = \Cbn{(\LTm\Sub{\LTmTwo}{\var})}$.
			
			\item \emph{Abstraction:} If $\LTm \Defeq \La{\varTwo}{\LTmThree}$ then $\Cbn{\LTm} = \La{\varTwo}{\Cbn{\LTmThree}}$ 
			(we 
			suppose without loss of generality 
			$\varTwo \notin \Fv{\LTmTwo} \cup \{\var\}$). 
			By 
			\Ih, $\Cbn{\LTmThree} \Sub{\Cbn{\LTmTwo}\!}{\var} = \Cbn{(\LTmThree \Sub{\LTmTwo}{\var})}$ and so $\Cbn{\LTm} \Sub{\Cbn{\LTmTwo}\!}{\var} = \La{\varTwo}{(\Cbn{\LTmThree} \Sub{\Cbn{\LTmTwo}\!}{\var})} = \La{\varTwo}{\Cbn{(\LTmThree \Sub{\LTmTwo}{\var})}} = \Cbn{(\LTm\Sub{\LTmTwo}{\var})}$.
			
			\item \emph{Application:} If $\LTm \Defeq \LTmThree\LTmFour$ then $\Cbn{\LTm} = \App{\Cbn{\LTmThree}}{\Bang{\Cbn{\LTmFour}}}$.
			By 
			\Ih, $\Cbn{\LTmThree} \Sub{\Cbn{\LTmTwo}}{\var} = \Cbn{(\LTmThree \Sub{\LTmTwo}{\var})}$ and $\Cbn{\LTmFour} \Sub{\Cbn{\LTmTwo}}{\var} \allowbreak= \Cbn{(\LTmFour \Sub{\LTmTwo}{\var})}$.
			Hence, 
			$
			\Cbn{\LTm} \Sub{\Cbn{\LTmTwo}}{\var} = \Cbn{\LTmThree} \Sub{\Cbn{\LTmTwo}}{\var}\Bang{(\Cbn{\LTmFour} \Sub{\Cbn{\LTmTwo}}{\var})} =
			\allowbreak \Cbn{(\LTmThree\Sub{\LTmTwo}{\var})}\Bang{({\Cbn{(\LTmFour\Sub{\LTmTwo}{\var})}} )} \allowbreak= \Cbn{(\LTm\Sub{\LTmTwo}{\var})}$.
			
			\item \emph{Operator:} If $\LTm \Defeq \op(\LTmThree_1, \dots, \LTmThree_n)$ then $\Cbn{\LTm} = \op(\Cbn{\LTmThree_1}, \dots, \Cbn{\LTmThree_n})$. 
			By 
			\Ih, $\Cbn{\LTmThree_i} \Sub{\Cbn{\LTmTwo}\!}{\var} = \Cbn{(\LTmThree_i \Sub{\LTmTwo}{\var})}$ for all $1 \leq i \leq n$, and so $\Cbn{\LTm} \Sub{\Cbn{\LTmTwo}\!}{\var} = \op(\Cbn{\LTmThree_1} \Sub{\Cbn{\LTmTwo}\!}{\var}, \dots, \Cbn{\LTmThree_n} \Sub{\Cbn{\LTmTwo}\!}{\var}) = \op(\Cbn{(\LTmThree_1 \Sub{\LTmTwo}{\var})}, \dots, \Cbn{(\LTmThree_n \Sub{\LTmTwo}{\var})}) = \Cbn{(\LTm\Sub{\LTmTwo}{\var})}$.
			
		\end{itemize}
		
		\smallskip
		\item 
		\begin{itemize}
			\item \emph{Variable:} If $\LTm$ is a variable then there are two subcases.
			If $\LTm \Defeq \var$ then 
			$\Cbv{\LTm} = \Bang{\var}$, so $\Cbv{\LTm} \Sub{\TmTwo}{\var} = \Bang{\TmTwo} = \Cbv{\LTmTwo} = \Cbv{(\LTm\Sub{\LTmTwo}{\var})}$.
			Otherwise $\LTm \Defeq \varTwo \neq \var$ and then $\Cbv{\LTm} = \Bang{\varTwo}$, hence $\Cbv{\LTm} \Sub{\TmTwo}{\var} = \Bang{\varTwo} = \Cbv{(\LTm\Sub{\LTmTwo}{\var})}$.
			
			\item \emph{Application:} If $\LTm \Defeq \LTmThree\LTmFour$ then there are two sub-cases:
			\begin{itemize}
				\item $\Cbv{\LTmThree} = \Bang{\TmThree}$ (\Ie $\LTmThree$ is a value) and then $\Cbv{\LTm} = \TmThree \Cbv{\LTmFour}$.
				By \Ih, $\Bangp{\TmThree \Sub{\TmTwo}{\var}} = \Cbv{\LTmThree} \Sub{\TmTwo}{\var} = \Cbv{(\LTmThree \Sub{\LTmTwo}{\var})}$ and $\Cbv{\LTmFour} \Sub{\TmTwo}{\var} \allowbreak= \Cbv{(\LTmFour \Sub{\LTmTwo}{\var})}$.
				Therefore, $\Cbv{\LTm} \Sub{\TmTwo}{\var} \allowbreak=  \TmThree \Sub{\TmTwo}{\var} \, \Cbv{\LTmFour} \Sub{\TmTwo}{\var} =\allowbreak \TmThree\Sub{\TmTwo}{\var}\allowbreak{\Cbv{(\LTmFour\Sub{\LTmTwo}{\var})}} \allowbreak= 
				\Cbv{(\LTmThree\Sub{\LTmTwo}{\var} \,\LTmFour\Sub{\LTmTwo}{\var})} \allowbreak=
				\Cbv{(\LTm\Sub{\LTmTwo}{\var})}$.
				
				\item $\Cbv{\LTmThree} \neq \Bang{\TmThree}$ for any $\TmThree \in \BangSet$ (\Ie $\LTmThree$ is not a value) and then $\Cbv{\LTm} = \App{\Der{\Cbv{\LTmThree}}}{\Cbv{\LTmFour}}$\!.
				By 
				\Ih, $\Cbv{\LTmThree} \Sub{\TmTwo}{\var} = \Cbv{(\LTmThree \Sub{\LTmTwo}{\var})}$ and $\Cbv{\LTmFour} \Sub{\TmTwo}{\var} \allowbreak= \Cbv{(\LTmFour \Sub{\LTmTwo}{\var})}$.
				So, $\Cbv{\LTm} \Sub{\TmTwo}{\var} \allowbreak= \App{ \Derp{\Cbv{\LTmThree} \Sub{\TmTwo}{\var}}}{\Cbv{\LTmFour} \Sub{\TmTwo}{\var}} =\allowbreak \App{ \Derp{\Cbv{\LTmThree\Sub{\LTmTwo}{\var}}}}\allowbreak{\Cbv{(\LTmFour\Sub{\LTmTwo}{\var})}} \allowbreak= \Cbv{(\LTm\Sub{\LTmTwo}{\var})}$.
			\end{itemize}
									
			\item \emph{Abstraction:} If $\LTm \Defeq \La{\!\varTwo}{\LTmThree}$ then $\Cbv{\LTm} \!=\! \Bang{(\La{\!\varTwo}{\Cbv{\LTmThree}})}$ 
			(suppose without loss of generality 
			$\varTwo \notin \Fv{\LTmTwo} \cup \{\var\}$). 
			By 
			\Ih\ $\Cbv{\LTmThree}\! \Sub{\TmTwo}{\var} = \Cbv{(\LTmThree \Sub{\LTmTwo}{\var})}$\!, 
			so $\Cbv{\LTm} \!\Sub{\TmTwo}{\var} = \Bangp{\La{\!\varTwo}{(\Cbv{\LTmThree} \!\Sub{\TmTwo}{\var})}} = \Bangp{\La{\!\varTwo}{\Cbv{(\LTmThree \Sub{\LTmTwo}{\var})}}} = \Cbv{(\LTm\Sub{\LTmTwo}{\var})}$\!.
			
			\item \emph{Operator:} If $\LTm \Defeq \op(\LTmThree_1, \dots, \LTmThree_n)$ then $\Cbv{\LTm} = \op(\Cbv{\LTmThree_1}, \dots, \Cbv{\LTmThree_n})$.
			By 
			\Ih, $\Cbv{\LTmThree_i} \Sub{\TmTwo}{\var} = \Cbv{(\LTmThree_i \Sub{\LTmTwo}{\var})}$ for all $1 \leq i \leq n$.
			So, $\Cbv{\LTm} \Sub{\TmTwo}{\var} \allowbreak= \op(\Cbv{\LTmThree_1} \Sub{\TmTwo}{\var}, \dots, \Cbv{\LTmThree_n} \Sub{\TmTwo}{\var}) =\allowbreak \op(\Cbv{(\LTmThree_1\Sub{\LTmTwo}{\var})}, \dots, \allowbreak\Cbv{(\LTmThree_n\Sub{\LTmTwo}{\var})}) \allowbreak= \Cbv{(\LTm\Sub{\LTmTwo}{\var})}$.
			\qed
		\end{itemize}
	\end{enumerate}
\end{proof}

Note that the hypothesis about $\LTmTwo$ in \Cref{lemma:substitution}.\ref{p:substitution-cbv} is fulfilled 
if and only if $\LTmTwo$ is a value.

\embedding*

\begin{proof}
	\begin{enumerate}
		\item \emph{Soundness:} We prove by induction on the term $\LTm \in \Lambda$ that if $\LTm \ToBeta \LTm'$ then $\Cbn{\LTm} \ToVal \Cbn{{\LTm'}}$. 
		According to the definition of $\LTm \ToBeta \LTm'$, there are the following cases:
		\begin{itemize}
			\item \emph{Root-step}, \Ie $\LTm \Defeq (\La{\var}{\LTmThree})\LTmTwo \Root{\beta} \LTmThree\Sub{\LTmTwo}{\var} \Eqdef {\LTm'}$: by \Cref{lemma:substitution}.\ref{p:substitution-cbn}, 
			$\Cbn{\LTm} = \App{\La{\var}{\Cbn{\LTmThree}}}{\Bang{\Cbn{\LTmTwo}}} \allowbreak\Root{\ValScript} \Cbn{\LTmThree} \Sub{\Cbn{\LTmTwo}\!}{\var} = \Cbn{{\LTm'}}$.
			
			\item  \emph{Abstraction}, \Ie $\LTm \Defeq \La{\var}{\LTmThree} \ToBeta \La{\var}{\LTmThree'} \Eqdef {\LTm'}$ with $\LTmThree \ToBeta \LTmThree'$: by 
			\Ih, $\Cbn{\LTmThree} \ToVal \Cbn{{\LTmThree'}}$, thus $\Cbn{\LTm} = \La{\var}{\Cbn{\LTmThree}} \ToVal \La{\var}{\Cbn{{\LTmThree'}}} = \Cbn{{\LTm'}}$.
			
			\item \emph{Application left}, \Ie $\LTm \Defeq \LTmThree\LTmTwo \ToBeta \LTmThree'\LTmTwo \Eqdef {\LTm'}$ with $\LTmThree \ToBeta \LTmThree'$:
			analogous to the previous case.
			
			\item \emph{Application right}, \Ie $\LTm \Defeq \LTmTwo\LTmThree \ToBeta \LTmTwo\LTmThree' \!\Eqdef \LTm'$ with $\LTmThree \ToBeta\! \LTmThree'$: 
			by 
			\Ih $\Cbn{\LTmThree} \ToVal \Cbn{{\LTmThree'}}$\!, so $\Cbn{\LTm} = \App{\Cbn{\LTmTwo}}{\Bang{{\Cbn{\LTmThree}}}} \!\ToVal \App {\Cbn{\LTmTwo}} {\Bang{{\Cbn{{\LTmThree'}}}}} \!= \Cbn{{\LTm'}}$\!.

			\item \emph{Operator}, \Ie $\LTm \Defeq \op(\dots, \LTmTwo, \dots) \ToBeta \op(\dots, \LTmTwo', \dots) \Eqdef \LTm'$ with $\LTmTwo \ToBeta \LTmTwo'$: analogous to the abstraction case.
		\end{itemize}
	
		The proof that if $\LTm \ToOp \LTm'$ then $\Cbn{\LTm} \ToOp \Cbn{{\LTm'}}$ is analogous, the only difference is in the root-step.
		\begin{itemize}
			\item \emph{Root-step}, \Ie $\LTm \Defeq \op(\LTmTwo_1, \dots, \LTmTwo_n) \Root{\op} \LTmThree \Eqdef {\LTm'}$: then, 
			$\Cbn{\LTm} = \op(\Cbn{\LTmTwo_1}, \dots, \Cbn{\LTmTwo_n}) \allowbreak\Root{\op} \Cbn{\LTmThree} = \Cbn{{\LTm'}}$.
		\end{itemize}
		
		\smallskip
		\emph{Completeness:} 
		We prove by induction on the term $\LTm \in \Lambda$ that if $\Cbn{\LTm} \ToVal \TmTwo$ then $\TmTwo = \Cbn{{\LTm'}}$ and $\LTm \ToBeta \LTm'$ for some $\lambda$-term $\LTm'$. 
		According to the definition of $\Cbn{\LTm} \ToVal \TmTwo$, there are the following cases:
		\begin{itemize}
			\item \emph{Root-step}, \Ie $\Cbn{\LTm} \Defeq \App{\La{\var}{\Cbn{\LTmThree}}}{\Bang{\Cbn{\LTmFour}}} \allowbreak\Root{\ValScript} \Cbn{\LTmThree} \Sub{\Cbn{\LTmFour}}{\var} \Eqdef \TmTwo$: 
			by \Cref{lemma:substitution}.\ref{p:substitution-cbn} $\TmTwo = \Cbn{(\LTmThree \Sub{\LTmFour}{\var})}$, so $\LTm = (\La{\var}{\LTmThree})\LTmFour \Root{\beta} \LTmThree\Sub{\LTmFour}{\var} \Eqdef {\LTm'}$ where $\Cbn{{\LTm'}} = \TmTwo$.
			
			\item  \emph{Abstraction}, \Ie $\Cbn{\LTm} \Defeq \La{\var}{\Cbn{\LTmThree}} \ToVal \La{\var}{\TmTwo'} \Eqdef \TmTwo$ with $\Cbn{\LTmThree} \ToVal \TmTwo'$: by 
			\Ih, there is a term $\LTmThree' \in \Lambda$ such that $\Cbn{{\LTmThree'}} = \TmTwo'$ and $\LTmThree \ToBeta \LTmThree'$, thus $\LTm = \La{\var}{\LTmThree} \ToBeta \La{\var}{\LTmThree'} \Eqdef {\LTm'}$ where $\Cbn{{\LTm'}} = \La{\var}\Cbn{{\LTmThree'}} = \TmTwo$.
			
			\item \emph{Application left}, \Ie $\Cbn{\LTm} \Defeq \App{\Cbn{\LTmThree}}{\Bang{{\Cbn{\LTmFour}}}} \ToVal \App{\TmTwo'}{\Bang{{\Cbn{\LTmFour}}}} \Eqdef \TmTwo$ with $\Cbn{\LTmThree} \ToVal \TmTwo'$: 
			analogously to above.
			
			\item \emph{Application right}, \Ie $\Cbn{\LTm} \!\Defeq \App{\Cbn{\LTmFour}}{\Bang{{\Cbn{\LTmThree}}}} \!\ToVal \App {\Cbn{\LTmFour}} \Bang{\TmTwo'} \!\Eqdef \TmTwo$ with $\Cbn{\LTmThree} \ToVal \TmTwo'$: by 
			\Ih, there is a term $\LTmThree' \in \Lambda$ such that $\Cbn{{\LTmThree'}} = \TmTwo'$ and $\LTmThree \ToBeta \LTmThree'$, so $\LTm = \LTmFour\LTmThree \ToBeta \LTmFour\LTmThree' \Eqdef \LTm'$ with $\Cbn{{\LTm'}} = \App{\Cbn{\LTmFour}}{\Bang{\Cbn{{\LTmThree'}}}} = \TmTwo$.
			
			\item \emph{Operator}, \Ie $\Cbn{\LTm} \Defeq \op(\dots, \Cbn{\LTmTwo}, \dots) \ToTot \op(\dots, \TmTwo', \dots) \Eqdef \TmTwo$ with $\Cbn{\LTmTwo} \ToTot \TmTwo'$: by
			\Ih, there is a term $\LTmTwo' \in \Lambda$ such that $\Cbn{{\LTmTwo'}} = \TmTwo'$ and $\LTmTwo \ToBeta \LTmTwo'$, thus $\LTm = \op(\dots, \LTmTwo, \dots) \ToBeta \op(\dots, \LTmTwo', \dots) \Eqdef \LTm'$ where $\Cbn{{\LTm'}} = \op(\dots, \Cbn{{\LTmTwo'}}, \dots) = \TmTwo$.
		\end{itemize}
		
		
		\item \emph{Soundness:} We prove a stronger statement by induction on the term $\LTm \in \Lambda$: if $\LTm \ToBetav \LTm'$ then $\Cbv{\LTm} \ToVal 
		\ToDer^= \Cbv{{\LTm'}}$ and moreover if $\LTm'$ is a value (\Ie an abstraction or a variable) and $\LTm$ is not then 
		$\Cbv{\LTm} \ToVal \Cbv{{\LTm'}}$.
		According to the definition of $\LTm \ToBetav \LTm'$, there are the following cases:
		\begin{itemize}
			\item \emph{Root-step}, \Ie $\LTm \Defeq (\La{\var}{\LTmThree})\LVal \Root{\Betav} \LTmThree\Sub{\LVal}{\var} \Eqdef \LTm'$ where $\LVal$ is a value, \Ie a variable or an abstraction: then, $\Cbv{(\La{\var}{\LTmThree})} = \Bang{\La{\var}{\Cbv{\LTmThree}}}$ and $\Cbv{\LVal} = \Bang{\TmTwo}$ for some $\TmTwo \in \BangSet$, hence $\Cbv{\LTm} = (\La{\var}{\Cbv{\LTmThree}}){\Bang{\TmTwo}}  \Root{\ValScript}\allowbreak \Cbv{\LTmThree} \Sub{\TmTwo}{\var} = \Cbv{{\LTm'}}$ by \Cref{lemma:substitution}.\ref{p:substitution-cbv}.
			
			\item  \emph{Abstraction}, \Ie $\LTm \Defeq \La{\var}{\LTmThree} \ToBetav \La{\var}{\LTmThree'} \Eqdef \LTm'$ with $\LTmThree \ToBetav \LTmThree'$: by 
			\Ih, $\Cbv{\LTmThree} \ToVal 
			\ToDer^= \Cbv{{\LTmThree'}}$; 
			therefore $\Cbv{\LTm} = \Bang{(\La{\var}{\Cbv{\LTmThree}})} \ToVal 
			\ToDer^= \Bang{(\La{\var}{\Cbv{{\LTmThree'}}})} = \Cbv{{\LTm'}}$ 
			where $\LTm$ is a value.
			
			\item \emph{Application left}, \Ie $\LTm \Defeq \LTmThree\LTmTwo \ToBetav \LTmThree'\LTmTwo \Eqdef \LTm'$ with $\LTmThree \ToBetav \LTmThree'$: by 
			\Ih, $\Cbv{\LTmThree} \ToVal 
			\ToDer^= \Cbv{{\LTmThree'}}$ and if $\LTmThree'$ is a value and $\LTmThree$ is not then 
			$\Cbv{\LTmThree} \ToVal \Cbv{{\LTmThree'}}$.
			There are two sub-cases:
			\begin{itemize}
				\item if $\Cbv{\LTmThree} = \Bang{\TmThree}$ for some $\TmThree \in \BangSet$, then
				$\Cbv{{\LTmThree'}} = \Bang{\TmThree'}$ for some $
				\TmThree' \in \BangSet$ such that $\TmThree \ToVal 
				\ToDer^= \TmThree'$; hence
				$\Cbv{\LTm} = \TmThree{{\Cbv{\LTmTwo}}} \ToVal 
				\ToDer^= \TmThree'{{\Cbv{\LTmTwo}}} = \Cbv{{\LTm'}}$;
				
				\item otherwise $\Cbv{\LTmThree} \neq \Bang{\TmThree}$ for any $\TmThree \in \BangSet$, thus $\LTmThree$ is not a value;
				either $\LTmThree'$ is a value and so $\Cbv{\LTmThree} \ToVal \Cbv{{\LTmThree'}} = \Bang{\TmThree'}$ for some $\TmThree' \in \BangSet$, hence $\Cbv{\LTm} = (\Der{\Cbv{\LTmThree}}){{\Cbv{\LTmTwo}}} \ToVal (\Der{\Bang{\TmThree'}}){{\Cbv{\LTmTwo}}} \ToDer \TmThree'{{\Cbv{\LTmTwo}}} = \Cbv{{\LTm'}}$; 
				or $\LTmThree'$ is not a value, thus $\Cbv{{\LTmThree'}} \neq \Bang{\TmThree'}$ for any $\TmThree' \in \BangSet$ and therefore $\Cbv{\LTm} = (\Der{\Cbv{\LTmThree}}){{\Cbv{\LTmTwo}}} \ToVal \ToDer^= (\Der{\Cbv{{\LTmThree'}}}){{\Cbv{\LTmTwo}}} = \Cbv{{\LTm'}}$.
			\end{itemize}
			In any case, $\LTm'$ is not a value.
			
			\item \emph{Application right}, \Ie $\LTm \Defeq \LTmTwo\LTmThree \ToBetav \LTmTwo\LTmThree' \!\Eqdef \LTm'$ with $\LTmThree \ToBetav\! \LTmThree'$: 
			analogous to the previous case.
			
			\item \emph{Operator}, \Ie $\LTm \Defeq \op(\dots, \LTmTwo, \dots) \ToBetav \op(\dots, \LTmTwo', \dots) \Eqdef \LTm'$ with $\LTmTwo \ToBetav \LTmTwo'$: analogous to the abstraction case, taking into account that $\LTm'$ is not a value.
		\end{itemize}
	
			The proof that if $\LTm \ToOp \LTm'$ then $\Cbv{\LTm} \ToOp \Cbv{{\LTm'}}$ is analogous, the only difference is in the root-step.
	\begin{itemize}
		\item \emph{Root-step}, \Ie $\LTm \Defeq \op(\LTmTwo_1, \dots, \LTmTwo_n) \Root{\op} \LTmThree \Eqdef {\LTm'}$: then, 
		$\Cbv{\LTm} = \op(\Cbv{\LTmTwo_1}, \dots, \Cbv{\LTmTwo_n}) \allowbreak\Root{\op} \Cbv{\LTmThree} = \Cbv{{\LTm'}}$.
	\end{itemize}

		\smallskip
		\emph{Completeness:} 
		We prove by induction on $\TmTwo_0 \in \BangSet$ that if $\Cbv{\LTm}  \ToVal \TmTwo_0 \ToDerNorm  \TmTwo$ then $\TmTwo_0 \ToDer^= \TmTwo = \Cbv{{\LTm'}}$ and $\LTm \ToBetav \LTm'$ for some term $\LTm' \in \Lambda$.
		According to the definition of $\Cbv{\LTm} \ToVal \TmTwo_0$, there are the following cases:
		\begin{itemize}
			\item \emph{Root-step}, \Ie $\Cbv{\LTm} \Defeq \App{\La{\var}{\TmThree}}{\Bang{\TmFour}} \Root{\ValScript} \TmThree\Sub{\TmFour}{\var} \Eqdef \TmTwo_0$.
			According to the definition of $\Cbv{(\cdot)}$, necessarily $\LTm = (\La{\var}{\LTmThree}){\LTmFour} \Root{\Betav} \LTmThree \Sub{\LTmFour}{\var} = \LTm'$ for some term $\LTmThree \in \Lambda$ and some value $\LTmFour$, with $\Cbv{\LTmThree} = {\TmThree}$ and $\Cbv{\LTmFour} = \Bang{\TmFour}$. 
			By substitution lemma (\Cref{lemma:substitution}.\ref{p:substitution-cbv}), $\Cbv{{\LTm'}} = \TmTwo_0$ and hence $\TmTwo_0$ is $\Derel$-normal, according to \Cref{rmk:cbv-val-normal}, thus $\TmTwo_0 = \TmTwo$.
			
			\item \emph{Abstraction}, \Ie $\Cbv{\LTm} \Defeq \La{\var}{\TmThree_0} \ToVal \La{\var}{\TmThree'} \Eqdef \TmTwo_0$ with $\TmThree_0 \ToVal \TmThree'$.
			This case is impossible because, according to the definition of $\Cbv{(\cdot)}$, there is no $\LTm \in \Lambda$ such that $\Cbv{\LTm}$ is an abstraction.
			
			
			\item \emph{Application left}, \Ie $\Cbv{\LTm} \Defeq \TmThree_0{\TmThree_1} \ToVal \App{\TmThree'}{\TmThree_1} \Eqdef \TmTwo$ with $\TmThree_0 \ToVal \TmThree'$.
			By \Cref{rmk:cbv-val-normal}, 
			$\Cbv{\LTm} = \App{\Der{\TmFive}}{\TmThree_1}$ for some $\ValScript$-normal $\TmFive \in \BangSet$ such that $\Der{\TmFive} \ToDer \TmThree_0$, and thus $\TmThree_0 = \Der{\TmFive_0}$ where $\TmFive \ToDer \TmFive_0$ (indeed $\TmFive = \Bang{\TmThree_0}$ is impossible because $\TmFive$ is $\ValScript$-normal), and hence $\TmThree' = \Der{\TmFive'}$ with $\TmFive_0 \ToVal \TmFive'$.
			So, $\LTm = \LTmFour\LTmFour_1$ for some $\lambda$-terms $\LTmFour$ and $\LTmFour_1$ such that $\Cbv{\LTmFour} = \TmFive$ and $\Cbv{\LTmFour_1} = \TmThree_1$ with $\Cbv{\LTmFour} \ToDer \TmFive_0 \ToVal \TmFive'$.
			By 
			\Ih, $\TmFive' = \Cbv{{\LTmFour'}}$ and $\LTmFour \ToBetav \LTmFour'$ for some $\lambda$-term $\LTmFour'$.
			Let $\LTm' \Defeq \LTmFour'\LTmFour_1$: then, $\LTm = \LTmFour\LTmFour_1 \ToBetav 
			\LTm'$ and $\Cbv{{\LTm'}} = \App{\Der{\Cbv{{\LTmFour'}}}}{\Cbv{\LTmFour_1}} = \TmTwo$.
			
			\item \emph{Application right}, \Ie $\TmTwo_0 \Defeq \App{\TmThree_1}{\TmThree_0} \ToVal \App{\TmThree_1}{\TmThree'} \Eqdef \TmTwo$ with $\TmThree_0 \ToVal \TmThree'$.
			By \Cref{rmk:cbv-val-normal}, necessarily $\Cbv{\LTm} = \App{\Der{\TmFive_1}}{\TmThree}$ for some $\ValScript$-normal $\TmThree, \TmFive_1 \in \BangSet$ such that $\Der{\TmFive_1} = \TmThree_1$ and $\TmThree \ToDer \TmThree_0$.
			So, $\LTm = \LTmFour_1\LTmFour$ for some $\lambda$-terms $\LTmFour_1$ and $\LTmFour$ such that $\Cbv{\LTmFour_1} = \TmFive_1$ and $\Cbv{\LTmFour} = \TmThree$ with $\Cbv{\LTmFour} \ToDer \TmThree_0 \ToVal \TmThree'$.
			By 
			\Ih, $\TmThree' = \Cbv{{\LTmFour'}}$ and $\LTmFour \ToBetav \LTmFour'$ for some $\lambda$-term $\LTmFour'$.
			Let $\LTm' \Defeq \LTmFour_1\LTmFour'$: then, $\LTm = \LTmFour_1\LTmFour \ToBetav 
			\LTm'$ and $\Cbv{{\LTm'}} = \App{\Der{\Cbv{\LTmFour_1}}}{\Cbv{{\LTmFour'}}} = \TmTwo$.
			
			\item \emph{Box}, \Ie $\TmTwo_0 \Defeq \Bang{\TmThree_0} \ToVal \Bang{\TmThree'} \Eqdef \TmTwo$ with $\TmThree_0 \ToVal \TmThree'$.
			According to \Cref{rmk:cbv-val-normal}, necessarily $\Cbv{\LTm} = \Bang{\TmThree}$ for some $\ValScript$-normal $\TmThree \in \BangSet$ such that $\TmThree \ToDer \TmThree_0$.
			So, $\LTm = \La{\var}{\LTmFour}$ (since $\Cbv{\LTm}$ is a box and $\Cbv{\var}$ is $\Derel$-normal) for some $\lambda$-term $\LTmFour$ such that $\TmThree = \La{\var}\Cbv{\LTmFour}$, and hence there are $\TmFive_0, \TmFive' \in \BangSet$ such that $\TmThree_0 = \La{\var}{\TmFive_0}$ and $\TmThree' = \La{\var}{\TmFive'}$  with $\Cbv{\LTmFour} \ToDer \TmFive_0 \ToVal \TmFive'$.
			By 
			\Ih, $\TmFive' = \Cbv{{\LTmFour'}}$ and $\LTmFour \ToBetav \LTmFour'$ for some $\lambda$-term $\LTmFour'$.
			Let $\LTm' \Defeq \La{\var}{\LTmFour'}$: then, $\LTm = \La{\var}{\LTmFour} \ToBetav 
			\LTm'$ and $\Cbv{{\LTm'}} = \Bangp{\La{\var}{\Cbv{{\LTmFour'}}}} = \TmTwo$.
			
			\item \emph{Operator}, \Ie $\Cbv{\LTm} \Defeq \op(\dots, \Cbn{\LTmTwo}, \dots) \ToTot \op(\dots, \TmTwo', \dots) \Eqdef \TmTwo$ with $\Cbv{\LTmTwo} \ToTot \TmTwo'$: by
			\Ih, there is a term $\LTmTwo' \in \Lambda$ such that $\Cbv{{\LTmTwo'}} = \TmTwo'$ and $\LTmTwo \ToBeta \LTmTwo'$, thus $\LTm = \op(\dots, \LTmTwo, \dots) \ToBeta \op(\dots, \LTmTwo', \dots) \Eqdef \LTm'$ where $\Cbv{{\LTm'}} = \op(\dots, \Cbv{{\LTmTwo'}}, \dots) = \TmTwo$.
			\qed
		\end{itemize}
		
	\end{enumerate}
\end{proof}

The targets of the \CbN translation $\Cbn{(\cdot)}$  and \CbV translation $\Cbv{(\cdot)}$
into the bang calculus can be \emph{characterized syntactically}, as described in \Cref{fig:targets}. 

\begin{remark}[Image of \CbN translation]
	\label{rmk:cbn-image}
	The \CbN translation $\Cbn{(\cdot)}$ 
	is a \emph{bijection} from the set $\Lambda$ of $\lambda$-terms to the subset $\Cbn{\BangSet\!}$ of $\BangSet$ defined in \Cref{fig:targets}:
	\begin{figure}[!t]
  \centering
  \scalebox{0.9}{\parbox{1.08\linewidth}{
  \begin{align*}
    \text{\emph{Target of \CbN translation into} $\BangSet$:}		&& \Tm, \TmTwo &\Coloneqq \var \mid \Tm \,{\Bang{\TmTwo}} \mid \La{\var}{\Tm}	\mid \op(\Tm_1, \dots, \Tm_k)		&&\textup{\!(set: } \Cbn{\BangSet\!}\textup{)}  
    \\
    \text{\emph{Target of \CbV translation into} $\BangSet$:}		&& \ImTmOne, \ImTmTwo &\Coloneqq \Bang{\ImVal} \mid (\Der{\ImTmOne}){\ImTmTwo} \mid \ImVal{\ImTmOne} \mid \op(\ImTmOne_1, \dots, \ImTmOne_n)				&&\textup{\!(set: }\Cbv{\BangSet\!}\textup{)} &
    \\&&\ImVal &\Coloneqq \var \mid \La{\var}{\ImTmOne} &&\textup{\!(set: }\Cbv{{\BangSet}_v}\textup{).}
  \end{align*}
  \vspace{-2\baselineskip}
  }}
  \caption{Targets of \CbN and \CbV translations of $\lambda$-terms into the bang calculus.}
  \label{fig:targets}
\end{figure}

	$\Cbn{\LTm} \in \Cbn{\BangSet\!}$ for any $\LTm \in \Lambda$, and conversely,
	for any $\Tm \in \Cbn{\BangSet\!}$, there is a \emph{unique} $\LTm \in \Lambda$ \mbox{such that $\Cbn{\Tm} = \LTm$.}
\end{remark}

By \Cref{thm:embedding}.\ref{p:embedding-cbn} and \Cref{rmk:cbn-image}, 
the image of the \CbN translation into the the bang calculus, \ie $\Cbn{\BangSet\!}$ endowed with reduction $\ToTot$, is \emph{isomorphic} to the pure \CbN $\lambda$-calculus.
In particular, $\Cbn{(\cdot)}$ preserves normal forms, forth and back.

\begin{corollary}[Preservations with respect to pure \CbN $\lambda$-calculus \cite{GuerrieriManzonetto18}]
	\label{cor:cbn-preservation}
	Let $\LTm, \LTmTwo \in \Lambda$.
	\begin{enumerate}
		\item\label{p:cbn-preservation-equational-theory} \emph{\CbN\ equational theory:} $\LTm \simeq_\beta \LTmTwo$ 
		if and only if $\Cbn{\LTm} \simeq_{\Tot} \Cbn{\LTmTwo}$.
		\item\label{p:cbn-preservation-normal} \emph{\CbN\ normal forms:} $\LTm$ is $\beta$-normal 
		if and only if $\Cbn{\LTm}$ is  $\Tot$-normal.
	\end{enumerate}
\end{corollary}
%

\begin{proof}\hfill
	\begin{enumerate}
%
		\item 
		If $\LTm \simeq_\beta \LTmTwo$ then $\LTm \ToBeta^* \LTmThree \MRevTo{\beta} \LTmTwo$ for some $\LTmThree \in \Lambda$, as $\ToBeta$ is confluent;
		by \Cref{thm:embedding}.\ref{p:embedding-cbn} (soundness), $\Cbn{\LTm} \ToVal^* \Cbn{\LTmThree} \MRevTo{\Tot} \Cbn{\LTmTwo}$ and so $\Cbn{\LTm} \simeq_{\Tot} \Cbn{\LTmTwo}$.
		
		Conversely, if $\Cbn{\LTm} \simeq_{\Tot} \Cbn{\LTmTwo}$ then $\Cbn{\LTm} \ToTot^* \TmThree \, \MRevTo{\Tot} \Cbn{\LTmTwo}$ for some $\TmThree \in \BangSet$, since $\ToTot$ is confluent;
		by \Cref{thm:embedding}.\ref{p:embedding-cbn} (completeness) and bijectivity of $\Cbn{(\cdot)}$ (\Cref{rmk:cbn-image}), $\LTm \ToBeta^* \LTmThree \MRevTo{\beta} \LTmTwo$ for some $\lambda$-term $\LTmThree$ such that $\Cbn{\LTmThree} = \TmThree$, and so $\LTm \simeq_\beta \LTmTwo$.

		\item Immediate consequence of \Cref{thm:embedding}.\ref{p:embedding-cbn}.
		\qed
	\end{enumerate}
\end{proof}

\smallskip
The correspondence between \CbV $\lambda$-calculus and bang calculus is slightly more delicate: \CbV translation 
$\Cbv{(\cdot)}$ gives a \emph{sound} and \emph{complete} embedding of $\ToBetav$ into $\ToVal\ToDerNorm$, but it is not complete with respect to 
$\ToTot$, 
see \Cref{ex:not-complete-cbv}
\begin{example}
	\label{ex:not-complete-cbv}
	Resuming \Cref{ex:simulation}, given the $\lambda$-term $\LTm = ((\La{\varThree}{\varThree})\var)\varTwo$, we have $\Cbv{\tm} \ToTot (\Der{\Bang{\var}})\Bang{\varTwo}$ 
	and there is no $\lambda$-term $\LTmTwo$ such that $\Cbv{\LTmTwo} = (\Der{\Bang{\var}})\Bang{\varTwo}$.
\end{example}

Nonetheless, as in the \CbN case, the \CbV translation preserves normal forms back and forth.
The proof requires a fine analysis of \CbV\ translation $\Cbv{(\cdot)}$.
First, 
we define two subsets $\Cbv{\BangSet}$ and $\Cbv{{\BangSet}_v}$ of $\BangSet$, see \Cref{fig:targets}. 
%
%
\begin{remark}[Image of \CbV\ translation]
	\label{rmk:cbv-image}
	If $\LTm \in \Lambda$ then $\Cbv{\LTm} \in \Cbv{\BangSet}$; in particular, if $\LVal \in \LambdaVal$ then $\Cbv{\LVal} = \Bang{\ImVal}$ for some $\ImVal \in \Cbv{{\BangSet}_v}$.
	However, $\Cbv{(\cdot)}$ is not surjective in $\Cbv{\BangSet}$: 
	there is no $\lambda$-term $\LTmTwo$ such that $\Cbv{\LTmTwo} = (\Der{\Bang{\var}})\Bang{\varTwo} \in \Cbv{\BangSet}$.
	Note that there is no term $\LTm \in \Lambda$ such that $\Cbv{\LTm} = \La{\var}{\var} = \Der{\!}$.
\end{remark}

We 
define a \emph{forgetful 
	map} $\InvV{(\cdot)} \colon \Cbv{\BangSet} \cup \Cbv{{\BangSet}_v} \to \Lambda$ 
from terms 
$\ImTmOne \!\in\! \Cbv{\BangSet\!}$ and $\ImVal \!\in\! \Cbv{{\BangSet}_v}$~into $\lambda$-terms,
just forgetting the $\oc$ and $\Der{}$ added by the \CbV translation.

\begin{equation*}
	\small
	\begin{gathered}
	\InvV{(\Bang{\ImVal})} = \InvV{\ImVal} 
	\quad 
	\InvV{((\Der{\ImTmOne}){\ImTmTwo})} = \InvV{\ImTmOne}\InvV{\ImTmTwo}
	\quad 
	\InvV{(\ImVal{\ImTmOne})} =  \InvV{\ImVal}\InvV{\ImTmOne}
	\quad
	\InvV{\op(\ImTmOne_, \dots, \ImTmOne_n)} = \op(\InvV{\ImTmOne_1}, \dots, \InvV{\ImTmOne_n});
  \\
	\InvV{\var} = \var 
	\qquad
	\InvV{(\La{\var}{\ImTmOne})} = \La{\var}{\InvV{\ImTmOne}}.
	\end{gathered}
\end{equation*}

Since $\Der{\!} = \La{\var}{\var} \notin \Cbv{\BangSet} \cup \Cbv{{\BangSet}_v}$, the forgetful function is well defined.

Thus, we can prove that $\Cbv{\BangSet\!}$ is the set of terms 
in the bang calculus that are reachable by $\Tot$-reduction or any $\op$-reduction from \CbV translations of $\lambda$-terms (\Ie for any $\LTm \in \Lambda$, if $\Cbv{\LTm} \To{\Tot,\op}^* \TmTwo$ then 
$\TmTwo \in \Cbv{\BangSet\!}$).
Also, $\Tot$-reduction on $\Cbv{\BangSet}$ is projected into $\Betav$-reduction on $\Lambda$ by the forgetful map $\InvV{(\cdot)}$, and \mbox{similarly for $\op$-reduction}.

%
\begin{lemma}[Properties of the forgetful 
	map $\InvV{(\cdot)}$]
	\label{lemma:forgetful}
	\hfill
	\begin{enumerate}
		\item\label{p:forgetful-inverse} \emph{Forgetful map 
			is a left-inverse of \CbV\ translation:} For every $\LTm \in \Lambda$,  $\InvV{\Cbv{\LTm}} = \LTm$.
		\item\label{p:forgetful-substitution} 
		\emph{Substitution:} $\ImTmOne \Sub{\ImVal}{\var} \!\in\! \Cbv{\BangSet\!}$ with $\InvV{(\ImTmOne \Sub{\ImVal}{\var})} = \InvV{\ImTmOne} \Sub{\InvV{\ImVal}\!}{\var}$, 
		for any $\ImTmOne \!\in\! \Cbv{\BangSet\!}$ and $\ImVal
		\!\in\! \Cbv{{\BangSet}_v}$.
		
		\item\label{p:forgetful-reduction} 
		\emph{$\Tot$ vs.~$\Betav$:}
		For all $\ImTmOne \in \Cbv{\BangSet\!}$ and $\Tm \in \BangSet$, if $\ImTmOne \ToTot \Tm$ then $\Tm \in \Cbv{\BangSet\!}$ and $\InvV{\ImTmOne} \ToBetav^= \InvV{\Tm}$.
		
		\item\label{p:forgetful-reduction-op} 
		\emph{$\op$ vs.~$\op$:}
		For all $\ImTmOne \in \Cbv{\BangSet\!}$ and $\Tm \in \BangSet$, if $\ImTmOne \ToOp \Tm$ then $\Tm \in \Cbv{\BangSet\!}$ and $\InvV{\ImTmOne} \ToOp^= \InvV{\Tm}$.
	\end{enumerate}
\end{lemma}

\begin{corollary}[Preservation with respect to pure \CbV $\lambda$-calculus]
	\label{cor:cbv-preservation}
	Let $\LTm, \LTmTwo \in \Lambda$. 
	\begin{enumerate}
		\item\label{p:cbv-preservation-equational-theory} \emph{\CbV equational theory:} $\LTm \simeq_{\Betav\!} \LTmTwo$ if and only if $\Cbv{\LTm} \simeq_{\Tot} \Cbv{\LTmTwo}$.
		\item\label{p:cbv-preservation-normal} \emph{\CbV normal forms:} $\LTm$ is $\Betav$-normal 
		if and only if $\Cbv{\LTm}$ is $\Tot$-normal.
	\end{enumerate}
\end{corollary}
%

\begin{proof}\hfill
	\begin{enumerate}
		\item If $\LTm \simeq_{\Betav} \LTmTwo$ then $\LTm \ToBetav^* \LTmThree \, \MRevTo{\Betav\!} \LTmTwo$ for some $\LTmThree \in \Lambda$, as $\ToBetav$ is confluent;
		by \Cref{thm:embedding}.\ref{p:embedding-cbv} (soundness), $\Cbv{\LTm} \ToTot^* \Cbv{\LTmThree} \MRevTo{\Tot} \Cbv{\LTmTwo}$ 
		and so $\Cbv{\LTm} \simeq_{\Tot} \Cbv{\LTmTwo}$.
		Conversely, if $\Cbv{\LTm} \simeq_{\Tot} \Cbv{\LTmTwo}$ then $\Cbv{\LTm} \ToTot^* \TmThree \,\MRevTo{\Tot} \Cbv{\LTmTwo}$ for some $\TmThree \in \BangSet$, since $\ToTot$ is confluent;
		by \Cref{rmk:cbv-image}, $\Cbv{\LTm}\!, \Cbv{\LTmTwo} \in \Cbv{\BangSet\!}$;
		thus, $\TmThree \in \Cbv{\BangSet\!}$ and $\InvV{\Cbv{\LTm}} \ToBetav^* \InvV{\TmThree} \, \MRevTo{\Betav\!} \InvV{\Cbv{\LTmTwo}}$ by \Cref{lemma:forgetful}.\ref{p:forgetful-reduction}, hence $\LTm = \InvV{\Cbv{\LTm}} \simeq_{\Betav} \InvV{\Cbv{\LTmTwo}} = \LTmTwo$ by \Cref{lemma:forgetful}.\ref{p:forgetful-inverse}.
		
		\item 
		If $\LTm$ is not $\Betav$-normal, then there is some $\TmTwo \in \BangSet$ such that $\Cbv{\LTm} \ToTot \TmTwo$ by \Cref{thm:embedding}.\ref{p:embedding-cbv} (soundness), hence $\Cbv{\LTm}$ is not $\Tot$-normal.
		
		Conversely, if $\Cbv{\LTm}$ is not $\Tot$-normal, then it is not $\ValScript$-normal according to \Cref{rmk:cbv-val-normal}, thus $\Cbv{\LTm} \ToVal\ToDerNorm  \TmTwo$ for some $\TmTwo \in \BangSet$.
		By \Cref{thm:embedding}.\ref{p:embedding-cbv} (completeness), $\LTm \ToBetav \LTmTwo$ with $\Cbv{\LTmTwo} = \TmTwo$.
		\qed
	\end{enumerate}
	
\end{proof}

\Cref{cor:cbv-preservation}.\ref{p:cbv-preservation-equational-theory} means that \CbV\ translation $\Cbv{(\cdot)}$\,---\,even if it is a sound but not complete embedding of $\Betav$-reduction into $\Tot$-reduction\,---\,is a sound and complete embedding of $\Betav$-equivalence into $\Tot$-equivalence.
Said differently, the non-completeness of $\Cbv{(\cdot)}$ with respect to $\Tot$-reduction is just a syntactic detail, the \CbV\ translation (via $\Tot$-equivalence) 
equates exactly the same as $\Betav$-equivalence.
\Cref{cor:cbv-preservation}.\ref{p:cbv-preservation-normal} says that $\Cbv{(\cdot)}$ preserves normal forms~(back~and~forth).

\section{Omitted proofs of 
\Cref{sec:ll_def}}

In this section, let  $\red$ be the contextual closure of a set of rules, and  $\R$  the set of all its redexes.

	\begin{definition}[Redex-preserving]		
		$\red$ is said  \textbf{redex-preserving} if      $T\nsred T'$  implies that
		$T$ is a redex if and only if $T'$ is a redex.
	\end{definition}
	
By \Cref{cor:bo_redex}, $\redbb\cup \ToOp$ is redex preserving.

\begin{lemma}[A sufficient condition for good least-level]\label{lem:good_redex}
If $\red$ is redex-preserving, then it has a good least-level.
\end{lemma}

\begin{proof}Given a redex-preserving reduction  $\red$, 
we  prove  
	\begin{enumerate}
		\item\label{p:ll-properties-monotone-bang}
		\emph{Monotonocity:}  $T \red S$ implies $\llev{T} \leq \llev{S}$.
		
		\item\label{p:ll-properties-invariance-bang}
		\emph{Internal invariance:}  $T \nllred S$ implies $\llev T = \llev{S}$. 
	\end{enumerate}

In both cases, the proof is by induction on $T \in \BangSet$.

\begin{enumerate} 
	\item \emph{Monotonicity.} Assume $T\red S$.
\begin{enumerate}
	\item\label{p:ll-properties-monotone-bang-root} If $T$ is a redex, then $\llev{T}=0  \leq \llev{S}$.

 \item If $T$ is  not a redex, and $S$ is a redex, then $T\sred S$ (by the assumption that $\red$ is redex-preserving), and so $\llev T=0=\llev S$.
 
\item  If neither $T$ nor $S$ is a redex,  we have the following  cases (note that 	  cases $T=(\lam x.P) \Bang{Q}$ and $T=\op(P_1, \dots, P_k)$ are included in \eqref{p:ll-properties-monotone-bang-root}).
\begin{itemize}
	\item $T=\lam x.P  \red \lam x.P'=S$. Then $P\red {P'}$, and we conclude by \ih, because $\llev {\lam x.P}=\llev P \leq \llev P' =\llev {\lam x.P'}$
	\item $T=!P  \red !P'=S$. Then $P\red P'$, and we conclude by \ih, because $\llev P \leq \llev {P'} $ implies 
	$\llev P \leq \llev { !P'}$.
	\item $T=PQ\red S$. Since $ T $ is not a redex,   either (i) $P\red P'$ and $PQ\red P'Q=S$, or (ii)  $Q\red Q'$ and $PQ\red PQ'=S$. In case (i), by \ih,  $\llev P \leq \llev {P'}$ and $\llev {P'Q}=\min \{\llev P, \llev Q\} \leq \min \{\llev P, \llev Q\}  =\llev {P'Q} $. Case  (ii) is similar. 
	Observe that in the definition of least level we use the fact that neither $T$ nor $S$ is a redex.
	
\end{itemize}
\end{enumerate}

\item \emph{Internal invariance.} Assume $T\nllred S$.

\begin{enumerate}
\item  If $T$ is a redex, then $S$ is a redex, because $\nllred \subset \nsred$. Therefore $\llev{T}=0  = \llev{S}$.

\item Otherwise,  if neither $T$ nor $S$ is a redex, and  we have the following  cases.
\begin{itemize}
		\item   $!T_1\nllred !  S_1$   and $T_1\nllred S_1$. We conclude by \ih.
	
	\item $\lam x. T_1 \nllred \lam x. S_1$  and $T_1\nllred S_1$. We conclude by \ih.
	
	\item $ T_1T_2\nllred S$ and  either (i) $T_1\red_{:k} S_1$ (and $S=S_1T_2$), 
	or (ii) $T_2\red_{:k} S_2$ (and $S=T_1S_2$). We have  $k>\llev{T_1T_2}=\min\{\llev{T_1}, \llev{T_2}\}$.
	
	  We examine case   $T_1\red_{:k} S_1$ (case (ii) is similar). \begin{itemize}
	  	\item If $\llev {T_1}\leq \llev{T_2}$, then   $T_1\nllred S_1$, and by \ih $\llev {T_1}= \llev{S_1}$.
	  Hence $\llev {S_1T_2}= \llev{S_1}=\llev{T_1}=\llev{T_1T_2} $
	  \item If $\llev {T_1}> \llev{T_2}$, since $\llev{S_1} \geq \llev {T_1}$ (by monotonicity), then $\llev{S_1T_2}=\llev{T_2}=\llev{T_1T_2}$.
	  \end{itemize}

\end{itemize}
Note that in the definition of least level, we use the fact that neither $T_1T_2$ nor $S_1S_2$ is a redex. 
\qed
\end{enumerate}

\end{enumerate}
\end{proof}

\Cref{lem:good_redex} and \Cref{cor:bo_redex} give


 \llproperties*
\section{Omitted proofs and lemmas of \Cref{sect:factorization}}

\paragraph{Factorization.}
Our proof of factorization follows and generalizes the approach used in \cite{AccattoliFaggianGuerrieri19} for the pure \CbN $\lambda$-calculus,
in particular the following characterization.

\begin{proposition}[Abstract factorization \cite{AccattoliFaggianGuerrieri19}]
	\label{prop:abstract-factorize}
	Let $\To{} \ = \ \essred \cup \nessred$ be a reduction over a set $\mathcal{A}$.
	Suppose that there are reductions $\Parallel$ and $\ParallelNotEss$ such that:
	\begin{itemize}
		\item \emph{Macro:} $\nessred \ \subseteq \ \ParallelNotEss \ \subseteq \nessred^*$;
		\item \emph{Merge:} For any $\tm, \tmTwo \in \mathcal{A}$, if $\tm \ParallelNotEss \!\cdot\! \essred \tmTwo$ then $\tm \Parallel \tmTwo$;
		\item \emph{Split:} For any $\tm, \tmTwo \in \mathcal{A}$, if $\tm \Parallel \tmTwo$ then $\tm \essred^* \!\cdot\! \ParallelNotEss \tmTwo$.
	\end{itemize}
	Then, $(\mathcal{A},\To{})$ $\esssym$-factorize: for any $\tm, \tmTwo \in \mathcal{A}$, if $\tm \To{}^* \tmTwo$ then $\tm \essred^* \!\cdot\! \nessred^* \tmTwo$.
\end{proposition}

Our goal is then to apply \Cref{prop:abstract-factorize} to the bang calculus, where $\To{} \ = \ \ToTot$ and $\essred \ = \ \llredBang$ and $\nessred \ = \ \nllredBang$. 
So, we have to identify reductions $\Parallel$ and $\ParallelNotEss$ in the bang calculus such that the properties macros, merge and split hold. 
The natural solution is to take:
\begin{itemize}
	\item $\Parallel \ = \ \ParallelBang$, the parallel version of $\ToTot$, which fires simultaneously a number of $\Tot$-redexes;
	\item $\ParallelNotEss \ = \ \ParallelNotLlBang$, the parallel version of $\nllredBang$, which fires simultaneously a number of $\Tot$-redexes that are not at minimal level.
\end{itemize}

Formally, \emph{parallel $\Tot$-reduction} $\ParallelBang$ is defined by the following rules:
\begin{center}
	\small
	\begin{prooftree}
		\infer0{\var \ParallelBang \var}
	\end{prooftree}
	\quad
	\begin{prooftree}
		\hypo{\Tm \ParallelBang \Tm'}
		\infer1{\La{\var}{\Tm} \ParallelBang \La{\var}{\Tm'}}
	\end{prooftree}
	\quad
	\begin{prooftree}
		\hypo{\Tm \ParallelBang \Tm'}
		\hypo{\TmTwo \ParallelBang \TmTwo'}
		\infer2{\Tm\TmTwo \ParallelBang \Tm'\TmTwo'}
	\end{prooftree}
	\quad
	\begin{prooftree}
		\hypo{\Tm \ParallelBang \Tm'}
		\infer1{\Bang{\Tm} \ParallelBang \Bang{\Tm'}}
	\end{prooftree}
	\\[5pt]
	\begin{prooftree}
		\hypo{\Tm \ParallelBang \Tm'}
		\hypo{\TmTwo \ParallelBang \TmTwo'}
		\infer2{(\La{\var}{\Tm})\Bang{\TmTwo} \ParallelBang \Tm'\Sub{\TmTwo'\!}{\var}}
	\end{prooftree}
	\quad
	\begin{prooftree}[separation=1.2em]
		\hypo{\Tm_1 \ParallelBang \Tm_1'}
		\hypo{\overset{k \in \Nat}{\dots}}
		\hypo{\Tm_k \ParallelBang \Tm_k'}
		\infer3{\op(\Tm_1, \dots, \Tm_k) \ParallelBang \op(\Tm_1', \dots, \Tm_k')}
	\end{prooftree}
\end{center}

To define $\ParallelNotLlBang$, we first introduce 
$\ParallelBangAt{n}$ (the parallel version of $\ToBangAt{n}$), which fires simultaneously a number of $\Tot$-redexes at level at least $n$ 
(and $\ParallelBangAt{\infty}$ does not reduce~any~$\Tot$-redex).
\begin{center}
	\small
	\begin{prooftree}
		\infer0{\var \ParallelBangAt{\infty} \var}
	\end{prooftree}
	\quad
	\begin{prooftree}
		\hypo{\Tm \ParallelBangAt{n} \Tm'}
		\infer1{\La{\var}{\Tm} \ParallelBangAt{n} \La{\var}{\Tm'}}
	\end{prooftree}
	\quad
	\begin{prooftree}[separation=1.2em]
		\hypo{\Tm \ParallelBangAt{m} \Tm'}
		\hypo{\TmTwo \ParallelBangAt{n} \TmTwo'}
		\infer2{\Tm\TmTwo \ParallelBangAt{\min\{m,n\}} \Tm'\TmTwo'}
	\end{prooftree}
	\quad
	\begin{prooftree}
		\hypo{\Tm \ParallelBangAt{n} \Tm'}
		\infer1{\Bang{\Tm} \ParallelBangAt{n\!+\!1} \Bang{\Tm'}}
	\end{prooftree}
	\\[5pt]
	\begin{prooftree}[separation=1.2em]
		\hypo{\Tm \ParallelBangAt{n} \Tm'}
		\hypo{\TmTwo \ParallelBangAt{m} \TmTwo'}
		\infer2{(\La{\var}{\Tm})\Bang{\TmTwo} \ParallelBangAt{0} \Tm'\Sub{\TmTwo'\!}{\var}}
	\end{prooftree}
	\quad
	\begin{prooftree}[separation=1.2em]
		\hypo{\Tm_1 \ParallelBangAt{n_1} \Tm_1'}
		\hypo{\overset{k \in \Nat}{\dots}}
		\hypo{\Tm_k \ParallelBangAt{n_k} \Tm_k'}
		\infer3{\op(\Tm_1, \dots, \Tm_k) \ParallelBangAt{1\!+\!\min\{n_1, \dots, n_k\}} \op(\Tm_1', \dots, \Tm_k')}
	\end{prooftree}
\end{center}

Note that $\Tm \ParallelBang \TmTwo$ if and only if $\Tm \ParallelBangAt{n} \TmTwo$ for some $n \in \Nat$; and $\Tm \ParallelBangAt{\infty} \TmTwo$ implies $\Tm = \TmTwo$.
The \emph{parallel internal $\Tot$-reduction} $\ParallelNotLlBang$ is then defined as:
\begin{center}
	\small
	$\Tm \ParallelNotLlBang \TmTwo$ \quad if $\Tm \ParallelBangAt{n} \TmTwo$ with $n = \infty$ or $n > \llBang{\Tm}$.
\end{center}

Clearly, $\ParallelBang$ and $\ParallelNotLlBang$ are reflexive, and $\nllredBang \ \subseteq \ \ParallelNotLlBang \ \subseteq \ \nllredBang^*$ (macro condition in \Cref{prop:abstract-factorize}).
To prove the merge property, we first prove a refined version of it ``by level''.

\begin{lemma}[Merge]
	\label{lemma:merge}
	In the bang calculus $(\BangSet, \ToTot)$:
	\begin{enumerate}
		\item\label{p:merge-level} \emph{Merge by level:}  If $\Tm \ParallelBangAt{n} \TmThree \ToBangAt{m} \TmTwo$ with $n > m$, then $\Tm \ParallelBang \TmTwo$.
		\item\label{p:merge-ll} \emph{Merge for least-level:} If $\Tm \ParallelNotLlBang \!\cdot\! \llredBang \TmTwo$, then $\Tm \ParallelBang \TmTwo$.
	\end{enumerate}
\end{lemma}


\begin{proof}\hfill
	\begin{enumerate}
		\item 	By induction on the definition of $\Tm \ParallelBangAt{n} \TmThree$.
		Consider the last rule of the derivation of $\Tm \ParallelBangAt{n} \TmThree$.
		It cannot conclude $\Tm = (\La{\var}{\Tm_0})\Bang{\Tm_1} \ParallelBangAt{0} \Tm_0 \Sub{\Tm_1}{\var} = \TmThree$ because otherwise $n = 0$, which contradicts the hypothesis $n > m \in \Nat$.
		Therefore, the only cases are:
		\begin{itemize}
			\item \emph{Variable}: $\Tm = \var \ParallelBangAt{\infty} \var = \TmThree$. 
			Then, there is no $\TmTwo$ such that $\TmThree \ToBangAt{m} \TmTwo$ for any $m \in \Nat$.
			
			\item \emph{Abstraction}: $\Tm = \La{\var}\Tm' \ParallelBangAt{n}  \La{\var} \TmThree' = \TmThree$ because $\Tm' \ParallelBangAt{n} \TmThree'$. 
			According to definition of $\TmThree \ToBangAt{m} \TmTwo$, by necessity $\TmTwo = \La{\var}{\TmTwo'}$ with $\TmThree' \ToBangAt{m} \TmTwo'$.
			By \Ih, $\Tm' \ParallelBang \TmTwo'$, thus
			\begin{align*}
			{
				\begin{prooftree}
				\hypo{\Tm' \ParallelBang \TmTwo'}
				\infer1{\Tm = \La{\var}{\Tm'} \ParallelBang \La{\var}{\TmTwo'} = \TmTwo}
				\end{prooftree}}
			\,.
			\end{align*}
			
			\item \emph{Box}: $\Tm = \Bang{\Tm'} \ParallelBangAt{n}  \Bang{\TmThree'} = \TmThree$ because $\Tm' \ParallelBangAt{n\!-\!1} \TmThree'$. 
			According to the definition of $\TmThree \ToBangAt{m} \TmTwo$, by necessity $\TmTwo = \Bang{\TmTwo'}$ with $\TmThree' \ToBangAt{m} \TmTwo'$.
			By \Ih, $\Tm' \ParallelBang \TmTwo'$, thus
			\begin{align*}
			{
				\begin{prooftree}
				\hypo{\Tm' \ParallelBang \TmTwo'}
				\infer1{\Tm = \Bang{\Tm'} \ParallelBang \Bang{\TmTwo'} = \TmTwo}
				\end{prooftree}}
			\,.
			\end{align*}
			
			\item \emph{Application}:
			\begin{align*}
			{
				\begin{prooftree}
				\hypo{\Tm_0 \ParallelBangAt{n_0} \TmThree_0}	
				\hypo{\Tm_1 \ParallelBangAt{n_1} \TmThree_1}
				\infer2{\Tm = \Tm_0 \Tm_1 \ParallelBangAt{n} \TmThree_0 \TmThree_1 = \TmThree}
				\end{prooftree}
			}
			\end{align*}
			where $n = \min\{n_0,n_1\}$.
			According to the definition of $\TmThree \ToBangAt{m} \TmTwo$, there are the following sub-cases:
			\begin{enumerate}
				\item $\TmThree = \TmThree_0\TmThree_1 \ToBangAt{m} \TmTwo_0\TmThree_1 = \TmTwo$ with $\TmThree_0 \ToBangAt{m} \TmTwo_0$; 
				since $m < n \leq n_0$, 
				by \Ih applied to $\Tm_0 \ParallelBangAt{n_0} \TmThree_0 \ToBangAt{m} \TmTwo_0$, we have $\Tm_0 \ParallelBang \TmTwo_0$, and so (as $\ParallelBangAt{n_1} \,\subseteq\, \ParallelBang$)
				\begin{align*}
				{\begin{prooftree}
					\hypo{\Tm_0 \ParallelBang \TmTwo_0}
					\hypo{\Tm_1 \ParallelBang \TmThree_1}
					\infer2{\Tm = \Tm_0\Tm_1 \ParallelBang \TmTwo_0\TmThree_1 = \TmTwo}
					\end{prooftree}}
				\,;
				\end{align*}
				\item $\TmThree = \TmThree_0\TmThree_1 \ToBangAt{m} \TmTwo_0\TmThree_1 = \TmTwo$ with $\TmThree_1 \ToBangAt{m} \TmTwo_1$; 
				analogous to the previous sub-case.
				
				\item $\TmThree = (\La{\var}{\TmThree_0'})\Bang{\TmThree_1'} \ToBangAt{0} \TmThree_0'\Sub{\TmThree_1'}{\var} = \TmTwo$ with $\TmThree_0 = \La{\var}\TmThree_0'$ and $\TmThree_1 = \Bang{\TmThree_1'}$ and $m = 0$; 
				as $0 < n \leq n_0,n_1$ then, according to the definition of $\Tm \ParallelBangAt{n} (\La{\var}{\TmThree_0'})\Bang{\TmThree_1'} = \TmThree$,
				\begin{align*}
				{
					\begin{prooftree}
					\hypo{\Tm_0 \ParallelBangAt{n_0} \TmThree_0'}
					\infer1{\La{\var}\Tm_0 \ParallelBangAt{n_0} \La{\var}\TmThree_0'}
					\hypo{\Tm_1 \ParallelBangAt{n_1\!-\!1} \TmThree_1'}
					\infer1{\Bang{\Tm_1} \ParallelBangAt{n_1}\Bang{\TmThree_1'}}
					\infer2{\Tm = (\La{\var}\Tm_0)\Tm_1 \ParallelBangAt{n} (\La{\var}{\TmThree_0'})\Bang{\TmThree_1'} = \TmThree}
					\end{prooftree}			}
				\end{align*}
				where $n = \min\{n_0, n_1\}$; therefore (as $\ParallelBangAt{k} \,\subseteq\, \ParallelBang$)
				\[
				\begin{prooftree}
				\hypo{\Tm_0 \ParallelBang \TmThree_0'}
				\hypo{\Tm_1 \ParallelBang \TmThree_1'}
				\infer2{\Tm = (\La{\var}{\Tm_0})\Bang{\Tm_1} \ParallelBang \TmThree_0'\Sub{\TmThree_1'}{\var} = \TmTwo}
				\end{prooftree}
				\]
			\end{enumerate}
		\end{itemize}
		
		\item Since $\Tm \ParallelNotLlBang \TmThree \llredBang \TmTwo$,  $\Tm \ParallelBangAt{n} \TmThree \ToTotInd{m} \TmTwo$ for some $n \in \Nat \cup \{\infty\}$
		and $m \in \Nat$ with $n > \llBang{\Tm}$ and $m = \llBang{\TmThree}$. 
		As $\ParallelNotLlBang \ \subseteq \ \nllredBang^*$ and $\nllredBang$
		cannot change the least-level 
		(\Cref{def:good}.\ref{p:ll-properties-invariance} and \Cref{prop:ll-properties}), $\llBang{\Tm} = \llBang{\TmThree}$ and so $n > m$. 
		By merging by level (\Cref{lemma:merge}.\ref{p:merge-level}), $\Tm \ParallelBang \TmTwo$.
		\qed
	\end{enumerate}
\end{proof}

The proof of the split property requires a further tool, to get the right induction hypothesis: the \emph{indexed parallel $\Tot$-reduction} $\ParallelBangInd{n}$ (not to be confused with $\ParallelBangAt{n}$), \ie $\ParallelBang$ equipped with a natural number $n$ which is, roughly, the number of $\Tot$-redexes reduced simultaneously by $\ParallelBang$.
The formal definition of $\ParallelBangInd{n}$ is (where $\Size{\Tm}_\var$ is the number of free occurrences of $\var$ in $\Tm$):
\begin{center}
	\small
	\begin{prooftree}
		\infer0{\var \ParallelBangInd{0} \var}
	\end{prooftree}
	\quad
	\begin{prooftree}
		\hypo{\Tm \ParallelBangInd{n} \Tm'}
		\infer1{\La{\var}{\Tm} \ParallelBangInd{n} \La{\var}{\Tm'}}
	\end{prooftree}
	\quad
	\begin{prooftree}
		\hypo{\Tm \ParallelBangInd{m} \Tm'}
		\hypo{\TmTwo \ParallelBangInd{n} \TmTwo'}
		\infer2{\Tm\TmTwo \ParallelBangIndLong{m+n} \Tm'\TmTwo'}
	\end{prooftree}
	\quad
	\begin{prooftree}
		\hypo{\Tm \ParallelBangInd{n} \Tm'}
		\infer1{\Bang{\Tm} \ParallelBangIndLong{n+1} \Bang{\Tm'}}
	\end{prooftree}
	\\[5pt]
	\begin{prooftree}
		\hypo{\Tm \ParallelBangInd{n} \Tm'}
		\hypo{\TmTwo \ParallelBangInd{m} \TmTwo'}
		\infer2{(\La{\var}{\Tm})\Bang{\TmTwo} \ParallelBangIndLong{n + \Size{\Tm'}_\var\cdot m + 1} \Tm'\Sub{\TmTwo'\!}{\var}}
	\end{prooftree}
	\quad
	\begin{prooftree}[separation=1.2em]
		\hypo{\Tm_1 \ParallelBangInd{n_1} \Tm_1'}
		\hypo{\overset{k \in \Nat}{\dots}}
		\hypo{\Tm_k \ParallelBangInd{n_k} \Tm_k'}
		\infer3{\op(\Tm_1, \dots, \Tm_k) \ParallelBangIndLong{n_1 + \dots + n_k} \op(\Tm_1', \dots, \Tm_k')}
	\end{prooftree}
\end{center}

The intuition behind the last clause is: $(\La{\var}{\Tm})\TmTwo$ reduces to $\Tm'\Sub{\TmTwo'\!}{\var}$ by:
\begin{enumerate}
	\item first reducing $(\La{\var}{\Tm})\TmTwo$ to $\Tm\Sub{\TmTwo}{\var}$ ($1$ step);
	\item then reducing in $\Tm\Sub{\TmTwo}{\var}$ the $n$ steps corresponding to the sequence $\Tm \ParallelBangInd{n} \Tm'$;
	\item finally reducing $\TmTwo$ to $\TmTwo'$ for every occurrence of $\var$ in $\Tm'$ replaced by $\TmTwo$, that is, $m$ 	steps for $\Size{\Tm'}_\var$ times, obtaining $\Tm'\Sub{\TmTwo'\!}{\var}$.
\end{enumerate}

\begin{lemma}[Split]
	\label{lemma:split}
	In the bang calculus $(\BangSet, \ToTot)$:
	\begin{enumerate}
		\item\label{p:split-indexed} \emph{Indexed split}: If $\Tm \ParallelBangInd{n} \TmTwo$ then $\Tm \ParallelNotLlBang \TmTwo$, or $n > 0 $ and $\Tm \llredBang \!\cdot\! \ParallelBangIndLong{n-1} \TmTwo$.
		\item\label{p:split-ll} \emph{Split for least-level}: If $\Tm \ParallelBang \TmTwo$, then $\Tm \llredBang^* \!\cdot\! \ParallelNotLlBang \TmTwo$.
	\end{enumerate}
\end{lemma}


\begin{proof}\hfill
	\begin{enumerate}
		\item By induction on the definition of $\Tm \ParallelBangInd{n} \TmTwo$. 
		Consider the last rule of the derivation of $\Tm \ParallelBangInd{n} \TmTwo$.
		We freely use the fact that if $\Tm \ParallelBangInd{n} \TmTwo$ then $\Tm \ParallelBang \TmTwo$. Cases:
		\begin{itemize}
			\item \emph{Variable}: $\Tm = \var \ParallelBangInd{0} \var = \TmTwo$. 
			Then, $\Tm = \var \ParallelNotLlBang \var = \TmTwo$  since $\var \ParallelBangInd{\infty} \var$.
			
			\item \emph{Abstraction}: $\Tm = \La \var{\Tm'} \ParallelBangInd n  \La {\var} {\TmTwo'} = \TmTwo$ because $\Tm' \ParallelBangInd n \TmTwo'$. It follows from the \Ih.
			
			\item \emph{Box}: $\Tm = \Bang{\Tm'} \ParallelBangInd{n}  \Bang{\TmTwo'} = \TmTwo$ because $\Tm' \ParallelBangInd{n} \TmTwo'$. It follows from the \Ih.
			
			\item \emph{Application}:
			\begin{align}
			\label{eq:app-parallel}
			{
			\begin{prooftree}
			\hypo{\TmThree \ParallelBangInd {n_1} \TmThree'}
			\hypo{\TmFive \ParallelBangInd {n_2} \TmFive'}
			\infer2{\Tm = \TmThree \TmFive \ParallelBangIndLong{n_1 + n_2} \TmThree' \TmFive' = \TmTwo}
			\end{prooftree}
			}
			\end{align}
			with $n = n_1 + n_2$.
			There are only two cases:
			\begin{itemize}
				\item either $\TmThree \TmFive \ParallelNotLlBang \TmThree' \TmFive'$, and then the claim holds;
				\item or $\TmThree\TmFive \not\ParallelNotLlBang \TmThree'\TmFive'$ and so (as $\TmThree\TmFive \ParallelBang \TmThree'\TmFive'$ with $\TmThree\TmFive \neq \TmThree'\TmFive'$) any derivation with conclusion $\TmThree \TmFive \ParallelBangAt{d} \TmThree' \TmFive'$ is such that $d =\llBang{\TmThree \TmFive} \in \Nat$.
				Let us rewrite derivation \eqref{eq:app-parallel} replacing $\ParallelBangInd{n}$ with $\ParallelBangAt{k}$: we have\footnote{This is possible because the inference rules for $\ParallelBangInd{n}$ and $\ParallelBangAt{k}$ are the same except for the way they manage their own indexes $n$ and $k$.}
				\[
				\begin{prooftree}
				\hypo{\TmThree \ParallelBangAt {d_{\TmThree}} \TmThree'}
				\hypo{\TmFive \ParallelBangAt{d_{\TmFive}} \TmFive'}
				\infer2{\Tm = \TmThree \TmFive \ParallelBangAt d \TmThree' \TmFive' = \TmTwo}
				\end{prooftree}
				\]
				where $d = \min\{d_\TmThree, d_\TmFive \}$.
				Thus, there are two sub-cases:
				\begin{enumerate}
					\item $d = d_\TmThree \leq d_\TmFive$ and then $d= \llBang{\TmThree \TmFive} \leq \llBang{\TmThree} \leq d_{\TmThree}=d $ (the first inequality holds by definition of $\llBang{\TmThree\TmFive}$), hence $\llBang \TmThree = d_{\TmThree}$; 
					we apply the \Ih\ to $\TmThree \ParallelBangInd{n_1} \TmThree'$ and we have that  $\TmThree \ParallelNotLlBang \TmThree'$, or $n_1>0$ and $\TmThree \llredBang \TmThree_1 \ParallelBangInd {n_1-1} \TmThree'$;
					but $\TmThree \ParallelNotLlBang \TmThree'$ is impossible because otherwise $\TmThree \TmFive \ParallelNotLlBang \TmThree' \TmFive'$ (as $d_\TmThree \leq d_\TmFive$);
					therefore, $n_1>0$ and $\TmThree \llredBang \TmThree_1 \ParallelBangInd{n_1-1} \TmThree'$,  so $n>0 $ and  $\Tm  = \TmThree \TmFive  \llredBang \TmThree_1 \TmFive  \ParallelBangIndLong {n_1-1+n_2} \TmThree'\TmFive' = \TmTwo$.
					  
					\item $d=d_{\TmFive} \leq d_{\TmThree}$ and then $d= \llBang{\TmThree \TmFive}\leq \deg \TmFive \leq d_{\TmFive} =d $, hence $\llBang{\TmFive} = 
					d_{\TmFive}$; we conclude analogously to thee previous sub-case.	
				\end{enumerate}
			\end{itemize}

			\item \emph{$\beta$ step}:
			\[
			\begin{prooftree}
			\hypo{\TmThree \ParallelBangInd {n_1} \TmThree'}
			\hypo{\TmFour \ParallelBangInd {n_2} \TmFour'}		
			\infer2{\Tm = (\La\var \TmThree)\TmFour \ParallelBangIndLong {n_1 + \SizeP{\TmThree'}\var \cdot n_2 +1} \TmThree'\Sub{\TmFour'}\var  = \TmTwo}
			\end{prooftree}
			\]
			With $n = n_1 + \SizeP{\TmThree'}\var \cdot n_2 +1 > 0$. We have $\Tm = (\La\var \TmThree)\TmFour \llredBang \TmThree \Sub{\TmFour}\var$ and by substitutivity of $\ParallelBangInd{n}$ (\Cref{partobind-subs}) $\TmThree\Sub{\TmFour}\var \ParallelBangIndLong{n_1 + \SizeP{\TmThree'}\var \cdot n_2} \TmThree' \Sub{\TmFour'}\var = \TmTwo$.
		\end{itemize}
	
		\item If $\Tm \ParallelBang \TmTwo$ then $\Tm \ParallelBangInd{n} \TmTwo$ for some $n \in \Nat$. We prove the statement by induction $n$. 
		By indexed split (\Cref{lemma:split}.\ref{p:split-indexed}), there are only two cases:
		\begin{itemize}
			\item \emph{$\Tm \ParallelNotLlBang \TmTwo$}. This is an instance of the statement (since $\llredBang^*$ is reflexive).
			
			\item $n>0$ and there exists $\TmFour$ such that $\Tm \llredBang \TmFour \ParallelBangInd{n-1} \TmTwo$. 
			By \Ih applied to $\TmFour \ParallelBangInd{n-1} \TmTwo$, there is $\TmThree$ such that $\TmFour \llredBang^* \TmThree \ParallelNotLlBang \TmTwo$, and so $\Tm \llredBang^* \TmThree \ParallelNotLlBang \TmTwo$.
		\qed
		\end{itemize}
	\end{enumerate}
\end{proof}

By \Cref{prop:abstract-factorize}, \Cref{lemma:merge,lemma:split}, 
least-level factorization of $\ToTot$ holds. 

\factorizebang*

%

\section{Omitted proofs of \Cref{sec:ll_modular}}

The proofs in this section  rely on  three ingredients:  the  properties of the  contextual 
closures---in particular \Cref{fact:shape}, as spelled out in \Cref{fact:isteps}---,the properties of  \emph{substitution},
which we recall below, and an assumption of good least-level.

A relation $\looparrowright$ on terms is \emph{substitutive} if 
\begin{equation}\tag{\textbf{substitutive}}
	R \looparrowright R' 
	\text{ implies } R \subs x Q \looparrowright R'\subs x Q.
\end{equation}
An obvious induction  on the shape of terms shows the   following (\cite{Barendregt84} pag. 54).
\begin{property}[Substitutive]\label{fact:subs} Let $\redc$ be  the contextual closure of $\rredc$.
	\begin{enumerate}
		\item\label{fact:subs-function} If $\rredc $ is substitutive then $\redc$ is 
		substitutive: ~ $T\redc T'$ implies $T \subs{x}{Q} \redc T' \subs{x}{Q}$.
		\item\label{fact:subs-argument} If $Q\redc Q'$ then $T\subs{x}{Q} \redc^* 
		T\subs{x}{Q'} $, always.
	\end{enumerate}
\end{property}
\medskip

We also rely on an assumption of good least-level, which we use to obtain  the following technical lemma.
\begin{lemma}\label{lem:nll_app}If $\red$ has a good least-level, then 
	$P Q \nllred  P' Q  \llred  P'' Q $ implies $P\nllred  P' \llred  P''$. Similarly, 
	$ Q  P\nllred  Q   P' \llred  Q  P''$ implies $P\nllred  P' \llred  P''$.
	
\end{lemma}
\begin{proof} Assume  $P\red_{:k}  P' \red_{:l}  P''$.  By assumption,  $l=\llev{P'Q}=\llev{P'}$, and so $ P' \llred  P''$.
	By internal invariance, $\llev{PQ}=l $, and so by assumption $k>l$.
	 By monotonicity,
	$\llev P \leq \llev{P'}=l <k$. Since $k>\llev {P}$, we have $P\nllred  P'$. 
\end{proof}

We recall that we often write $\Root{}$ to indicate the step $\red$ which is obtained by \emph{empty contextual closure}.

\subsection{ $\BangSetOp$: least-level Factorization, Modularly.}
Consider a calculus $(\BangSetOp, \red \eq \redbb\cup \redc)$, where $\redc$ is a new reduction added to $\bbeta$. \Cref{thm:modular}
 {states} that the  compound system $\redbb\cup \redc$ satisfies  least-level  factorization 
if $\F\llredbb\nllredbb$, $\F\llredc\nllredc$, 
and  the two linear swaps hold. We have already proved that  if $\red$ has a  good least-level, then  $\F\llredbb\nllredbb$ always hold. We now show that to verify    the linear swaps  
reduces to a single simple test, leading to \Cref{prop:test_ll}.

First, we observe that each  linear swap condition can be tested by considering for the least-level step  only $\mapsto$, that is, only the closure of $\mapsto$ under \emph{empty} context.
This is expressed in the following lemma, where we  include also a useful variant.

\begin{lemma}[Root linear swaps]\label{l:ll_swaps}Let  $\reda, \redc$ be the contextual closure of rules $\rreda,\rredc$, and  assume $\reda\cup \redc$ to have a good least-level.

\begin{enumerate}
	\item 	$\nllredx{\alpha} \cdot \rredc \subseteq  {\llredx{\gamma}} \cdot \reda^* $ implies 
	$\nllredx{\alpha} \cdot \llredx{\gamma} \subseteq  {\llredx{\gamma}} \cdot \reda^* $. 
	
	\item Similarly,  $\nllredx{\alpha} \cdot \rredc \subseteq  {\llredx{\gamma}} \cdot \reda^= $ implies 
	$\nllredx{\alpha} \cdot \llredx{\gamma} \subseteq  {\llredx{\gamma}} \cdot \reda^= $. 
\end{enumerate}
\end{lemma}
\begin{proof}Assume $M \nllredx{\alpha} U \llredx{\gamma}N$. 
	If $U$ is the redex, the claim holds by assumption. Otherwise,
	we prove  $M{\llredx{\gamma}} \cdot \reda^* N $, by induction on the structure of $ U $.  Observe that  both $M$ and $N$ have the same shape as $U$ (by Property \ref{fact:shape} ).
	\begin{itemize}
		\item $U=U_1U_2$ (hence $M=M_1M_2$ and $N=N_1N_2$). We   have  two cases.   
		\begin{enumerate}
			\item Case  $U_1 \llredc N_1$.  By \Cref{fact:isteps}, either $M_1\reda U_1$ or $M_2\reda U_2$.
			
		\begin{enumerate}
			\item 	 Assume $M:=M_1M_2 \nllreda U_1 M_2\llredc N_1 M_2:=N$. By  \Cref{lem:nll_app}, we have $M_1 \nllreda U_1 \llredc N_1 $,
			and we conclude by \ih.
			
			\item  Assume $M:=U_1M_2 \nllreda U_1 U_2\llredc N_1 U_2:=N$.  Then $U_1M_2 \llredc N_1M_2\reda N_1U_2$.
		\end{enumerate}
			\item Case $U_2 \llredc N_2$. Similar to the above.	 
		\end{enumerate}
		\item $U=\lam x.U_0$ (hence $M=\lam x. M_0$ and $N=\lam x. N_0$). We conclude by  \ih.
		\item $U=!U_0$ (hence $M=! M_0$ and $N=! N_0$). We conclude by \ih.
	\end{itemize}
	
\end{proof}

Since we    study $ \redb\cup \redc$, one of the linear swap is 	$\nllredx{\gamma} \cdot \llredbb \subseteq  {\llredbb} \cdot \redc^*$. We show that, whatever is    $\redc$, it  linearly swaps after $\llredbb$ as soon as $\rredc$  is \emph{substitutive}. 

\begin{lemma}[Swap with $\llredbb$] \label{l:swap_after_b} 
	Let  $\redc$ be  the contextual closure of rule $\rredc$, and  assume $\redbb\cup \redc$ has    good least level.
	
	If  	  $\rredc$ is   substitutive, then 
	$\nllredx{\gamma} \cdot \llredbb \subseteq  {\llredbb} \cdot \redc^* $  always holds. 
\end{lemma}
\begin{proof} We prove $ \nllredx{\gamma} \cdot  \mapsto_{\bbeta} ~\subseteq  ~ {\llredbb} \cdot \redc^* $, and conclude by  \Cref{l:ll_swaps}.
	 
Assume  $M \nllredx{\gamma} (\lam x.P) !Q \rredbb P \subs x Q$. We want to prove   $  M{\llredbb} \cdot \redc^* P \subs x Q$.  By \Cref{fact:isteps},  $M=M_1M_2$ and {either $M_1= \lam x.M_P \redc  (\lam x.P)$ or $M_2= !M_Q \redc !Q$. }
\begin{itemize}
	\item 	In the first case,  $M=(\lam x.M_P)!Q$, with $M_P \redc P$. Hence $M=(\lam x.M_P)!Q \rredbb M_P \subs x Q $ and we conclude by substitutivity of $\redc$ (Property \ref{fact:subs}, point 1.).
	\item In the second case, $M=(\lam x.P)!M_Q$ with $M_Q \redc Q$. Hence $M=(\lam x.P)!M_Q \rredbb P \subs x {M_Q} $, and we conclude by Property \ref{fact:subs}, point 2.

\end{itemize}	
\end{proof}

Summing up, since surface factorization for $ \bbeta$ is known, we obtain the following compact  test for least-level factorization in extensions of  $\Lambda^!$.

\testll*

%

\subsection{Modular factorization for CbV and CbN} A modularity result similar to \Cref{prop:test_ll}  can be established for CbN and CbV.  Moreover, there is nothing special with least-level: leftmost works equally well.

\begin{proposition}[A test for  CbN least-level factorization] Let $\redc$ be the contextual closure of $\rredc $, and assume 	
	$\red \eq( \redb \cup \redc)$ has a \good least-level. 
 	The union  	$\redb \cup \redc$  satisfies CbN least-level factorization if:
	\begin{enumerate}
		\item \emph{$\ll$-factorization of $\redc$}: $\F{\llredx{\gamma}}{\nllredx{\gamma}}$.
		\item \emph{Substitutivity}: $\rredc$ is  substitutive.
		\item \emph{Root linear swap}: $\nllredx{\beta} \cdot \rredc  \ \subseteq \ \llredx{\gamma} \cdot\redb^* $.
	\end{enumerate}
	
\end{proposition}

\begin{proposition}[A test for  CbV least-level factorization] Let $\redc$ be the contextual closure of $\rredc$,  and assume 	
	$\red \eq( \redbv \cup \redc)$ has a \good least-level. 
	The union  	$\redbv \cup \redc$  satisfies CbV least-level factorization if:
	\begin{enumerate}
		\item \emph{$\ll$-factorization of $\redc$}: $\F{\llredx{\gamma}}{\nllredx{\gamma}}$.
		\item \emph{Substitutivity}: $\rredc$ is  substitutive.
		\item \emph{Root linear swap}: $\nllredx{\betav} \cdot \rredc  \ \subseteq \ \llredx{\gamma} \cdot\redbv^* $.
	\end{enumerate}
	
\end{proposition}

}

\end{document}